\documentclass[12pt]{amsart}

\input{self-test.sty}
\title[A non-robust quantum correlation self-test]{A non-Robust quantum correlation self-test}
\author{Ranyiliu Chen}
\author{Yuming Zhao}
\date{}

\address{Quantum Science Center of Guangdong-Hong Kong-Macao Greater Bay Area, Shenzhen, China}
\email{chenranyiliu@quantumsc.cn}

\address{QMATH, Department of Mathematical Sciences, University of Copenhagen, Denmark}
\email{yuming@math.ku.dk}

\begin{document}

\begin{abstract}
  We show that exact self-testing of a quantum correlation does not imply robust self-testing. Man\v{c}inska and Schmidt previously established such a separation for non-local games, but their obstruction to robustness arises from two distinct optimal correlations, with the self-tested one remaining robust. We establish the separation for a single correlation.

  Our proof is operator-algebraic. The main technical contribution is a general construction that associates a synchronous binary correlation to any finitely presented $C^*$-algebra with unitary generators and real-coefficient relations. The implementing states of the correlation correspond exactly to the tracial states on the algebra. We apply this construction to an algebra with a unique finite-dimensional tracial state and a distinct amenable tracial state. By a characterization established by Zhao and independently by Kar, the resulting correlation is a self-test but not a robust self-test.
\end{abstract}
\maketitle
\vspace{-0.2cm}

\paragraph{\textbf{Disclosure of AI usage}}
ChatGPT (5.6 sol ultra) proposed an initial solution to the research
problem formulated by the authors. The authors subsequently verified,
simplified, and further developed the arguments. These refinements made the approach more amenable to generalization and brought out its broader algebraic significance. The authors take full
responsibility for all statements and proofs in this manuscript. 

\vspace{-0.2cm}

\tableofcontents
\section{Introduction}

In a device-independent quantum protocol, a classical verifier interacts with two non-communicating quantum players and wishes to certify that they share a prescribed entangled state and perform commanded measurements. Since the verifier has no access to the players' devices, this conclusion must be drawn entirely from their classical answers to classical questions. \emph{Self-testing} provides such a certification: for certain observed input-output statistics, every quantum strategy producing these statistics must contain the same ideal strategy, up to local isometries and an auxiliary system \cite{MayersYao04}.

More concretely, in a bipartite Bell scenario with question sets $X,Y$ and answer sets $A,B$, a \emph{quantum strategy} $\mcS$ consists of a state $ \ket{\psi}\in\mcH_A\otimes\mcH_B$ and local measurements $\{P_{x,a}\}_{a\in A}\subset \msB(\mcH_A), \{Q_{y,b}\}_{b\in B}\subset\msB(\mcH_B)$. It produces a \emph{quantum correlation}
\begin{equation*}
  p_{\mcS}(a,b|x,y) = \bra{\psi}P_{x,a}\otimes Q_{y,b}\ket{\psi}.
\end{equation*}
%which describes the probability that Alice and Bob return answers $a\in A$and $b\in B$ upon receiving questions $x\in X$ and $y\in Y$.
Throughout this paper, all quantum strategies are \emph{finite-dimensional}, and the resulting set of quantum correlations is denoted by $C_q(X,Y,A,B)$.

Self-testing can be formulated either for correlations or for non-local games. A correlation $p$ is a self-test if every quantum strategy producing $p$ is equivalent to one fixed ideal strategy. In a non-local game, the verifier samples questions and decides whether the players win from their questions and answers. Such a game is a self-test if every optimal quantum strategy is equivalent to one fixed ideal optimal strategy.
%Thus, correlation self-testing uses the complete input-outputstatistics, whereas game self-testing uses only the winning probability obtained from those statistics. The roots of self-testing can be traced back to foundational work of Summers and Werner \cite{SW87,SW88}, Popescu and Rohrlich \cite{PR92},and Tsirelson \cite{Tsi93}, and the notion was later introduced by Mayers and Yao \cite{MayersYao04} from a cryptographic perspective.
Self-testing in both settings are important tools for device-independent quantum certification, with applications to device-independent cryptography~\cite{VV14,MS16}, verifiable quantum computation and delegation~\cite{RUV13,CGJV19,BMZ24}, and quantum complexity theory~\cite{NV18,JNVWY22,MNY22,NZ23}; see \cite{SupicBowles20} for a survey.

For these applications, exact self-testing is not sufficient. A physical experiment involves noise and imperfections, and the verifier can only estimate the players' input-output statistics from finitely many rounds. Consequently, the observed correlation or winning probability can only be expected to be close to its ideal value. To obtain a sound certification protocol, one needs a \emph{robust} self-test: every strategy producing a correlation close to the ideal correlation, or achieving a winning probability close to the optimal value, must itself be close to the ideal strategy, up to local isometries.

This leads to the following natural question:
\begin{question}\label{Q1}
  Does exact self-testing automatically provide such a soundness guarantee? In other words, must every self-test be robust?
\end{question}

There is no immediate compactness argument for this implication. Although every quantum strategy considered here is finite-dimensional, there is no uniform bound on its dimension. In fact, the set of quantum correlations $C_q(X,Y,A,B)$ is not closed in general \cite{Slofstra19,DPP19}, and we denote its closure by $C_{qa}(X,Y,A,B)$ whose elements are called \emph{quantum-approximate correlations}.

At the level of non-local games, Man\v{c}inska and Schmidt answered \Cref{Q1} negatively \cite{MancinskaSchmidt23}. Their construction starts from two non-local games $G_1$ and $G_2$ and forms the game $G_1\lor G_2$, in which the verifier sends each player one question from each component game, and each player chooses which question to answer. The players win $G_1\lor G_2$ if they choose the same component game and win that game.

The game $G_1$ has a perfect quantum-approximate correlation $p_1\in C_{qa}$, but no perfect finite-dimensional strategy. In contrast, $G_2$ has a perfect quantum correlation $p_2\in C_q$ and self-tests a corresponding quantum strategy $\wtd{\mcS}$. Every perfect finite-dimensional strategy for $G_1\lor G_2$ must choose $G_2$ and play it perfectly, so $G_1\lor G_2$ self-tests $\wtd{\mcS}$. On the other hand, finite-dimensional strategies approximating $p_1$ give near-perfect strategies for $G_1\lor G_2$, but these strategies do not approach the ideal strategy $\wtd{\mcS}$. Hence $G_1\lor G_2$ is not a robust self-test.

The crucial point is that, in $G_1\lor G_2$, the players' answers reveal which branch they have chosen. There is therefore a natural way to make a certification protocol based on $G_1\lor G_2$ still sound: the verifier rejects whenever the players choose $G_1$ and otherwise applies the $G_2$-test, which may be chosen to robustly self-test $\wtd{\mcS}$. Although this changes the formal acceptance predicate, it uses only information already contained in the players' answers and removes the obstruction to soundness caused by the strategies converging to $p_1$.

Thus, although this example answers \Cref{Q1} negatively for non-local games, it does not fully settle the soundness question motivating it.
%The failure of robustness comes from the other distinguishable correlation $p_1$, rather than from correlations approaching the self-tested correlation $p_2$. 
It leaves open whether exact self-testing of a \emph{single} quantum correlation always yields a sound device-independent protocol. In other words:
\begin{question}\label{Q2}
  If a quantum correlation is a self-test, must it be a robust self-test?
\end{question}

We answer this sharper question in the negative.

\begin{theorem}\label{thm:main}
  There exists a synchronous quantum correlation $\wtd{p}$ such that $\wtd{p}$ is a self-test but not a robust self-test.
\end{theorem}

A quantum correlation $p\in C_q(X,X,A,A)$ is said to be \emph{synchronous} if $p(a,b|x,x)=0$ for every $x\in X$ and distinct $a,b\in A$. We write $C_q^s(X,A)$ for the set of such correlations. Despite this simple definition, synchronous correlations admit a rich operator-algebraic description: correlations arising from different models of entanglement are characterized by different kinds of \emph{tracial states}.

More concretely, let $\PVM{X,A}$ denote the universal $C^*$-algebra generated by ``abstract" PVMs $\{e_{x,a}:a\in A\}, x\in X$. A quantum correlation $p\in C_q(X,X,A,A)$ is synchronous if and only if there is a finite-dimensional tracial state $\tau$ on $\PVM{X,A}$ \emph{implementing} $p$, in the sense that
\begin{equation*}
  p(a,b|x,y)=\tau(e_{x,a}e_{y,b})
\end{equation*}
for all $x,y\in X$ and $a,b\in A$. Here, a tracial state is called finite-dimensional if its GNS representation is finite-dimensional. Moreover, synchronous quantum-approximate correlations correlations are characterized by \emph{amenable} tracial states \cite{PSSTW16,KPS18,HMPS19}.

This tracial framework has become an important tool for connecting quantum interactive proofs with operator algebras, including in the developments surrounding $\mathrm{MIP}^*=\mathrm{RE}$ and the resulting negative solution to Connes' embedding problem \cite{JNVWY22,MNY22,NZ23}. It also provides the starting point for our construction.

\subsection{Proof approach and technical contributions}
Our proof is based precisely on this tracial-state characterization.
Tracial states not only describe synchronous correlations, but also
determine when such a correlation is a self-test or a robust self-test.
The following operator-algebraic characterizations were first established
in \cite{Zhao24} and independently proved in \cite{Kar25}. Suppose
$p\in C_q^s(X,A)$ is an \emph{extreme} point in
$C_q(X,X,A,A)$. Then:
\begin{enumerate}[(a)]
    \item $p$ is a self-test if and only if there is a
    \emph{unique finite-dimensional} tracial state on $\PVM{X,A}$
    implementing $p$.
    \item $p$ is a robust self-test if and only if there is a
    \emph{unique amenable} tracial state on $\PVM{X,A}$
    implementing $p$.
\end{enumerate}

Note that every finite-dimensional tracial state is amenable. To prove
\Cref{thm:main}, we therefore seek an extreme synchronous correlation
that has a unique finite-dimensional implementing tracial state but
also a distinct amenable implementing tracial state.

We start by constructing a $C^*$-algebra
$\mcA_{\mathrm{rot}}$
with analogous uniqueness and separation properties for its tracial
states. This algebra admits a simple finite presentation by unitary
generators and relations with real coefficients, and has the decomposition
\begin{equation*}
    \mcA_{\mathrm{rot}}\cong\C\oplus\mcA_\infty.
\end{equation*}
Here $\mcA_\infty$ contains irrational rotation relations among
unitaries. It has no finite-dimensional tracial states, but carries
an amenable tracial state $\tau_\infty$. Consequently, the character
$\tau_0(c\oplus b):=c$ is the unique finite-dimensional tracial state
on $\mcA_{\mathrm{rot}}$, while
$\tau_1(c\oplus b):=\tau_\infty(b)$ is a distinct amenable tracial
state. Thus the desired uniqueness and separation are already present
at the level of this algebra. We give its presentation and establish
these properties in \Cref{sec:rotation-example}.

To transfer these properties to the implementing tracial states for
synchronous correlations, we establish the following general
construction, proved in \Cref{thm:projection-construction}. This is
our main technical contribution.

\begin{theorem}\label{thm:projection-construction-intro}
    Suppose $\mcA$ admits a tracial state and has a presentation by
    finitely many unitary generators and finitely many $*$-polynomial
    relations with real coefficients. There exist a positive integer
    $N$, a finite set $X$, and projections $P_x\in M_N(\mcA)$,
    $x\in X$, with the following properties.
    \begin{enumerate}
        \item[(P1)] The projections $P_x$ generate $M_N(\mcA)$.
        Thus the $*$-homomorphism
        $\lambda:\PVM{X,\Z_2}\arr M_N(\mcA)$ sending
        $e_{x,1}\mapsto P_x$ is surjective.

        \item[(P2)] For any $x,y\in X$, the value
        $(\tr_N\otimes\tau)(P_xP_y)$ is independent of the choice
        of tracial state $\tau$ on $\mcA$. Denote this value by
        $p_{x,y}$, where $\tr_N$ is the normalized trace on
        $M_N(\C)$.

        \item[(P3)] Every tracial state $\varphi$ on
        $\PVM{X,\Z_2}$ satisfying
        $\varphi(e_{x,1}e_{y,1})=p_{x,y}$ for all $x,y\in X$
        factors through $\lambda$, i.e.,
        $\varphi=(\tr_N\otimes\tau_\varphi)\circ\lambda$ for a
        unique tracial state $\tau_\varphi$ on $\mcA$.
    \end{enumerate}
\end{theorem}

By (P2), the value
\begin{equation}
    p_{\mcA}(a,b\mid x,y)
    :=(\tr_N\otimes\tau)(P_{x,a}P_{y,b})
    \label{eq:intro-associated-correlation}
\end{equation}
is independent of the choice of tracial state $\tau$ on $\mcA$,
where $P_{x,1}:=P_x$ and $P_{x,0}:=\Id-P_x$. This defines the
associated synchronous binary correlation
$p_{\mcA}\in C_{qc}^s(X,\Z_2)$. Here $C_{qc}$ denotes the
commuting-operator correlation set. Property (P3) further
determines all implementing tracial states and implies extremality.
We prove the following correspondence in
\Cref{prop:app-nonclosure-selftest}.

\begin{proposition}\label{prop:trace-correspondence-intro}
    The correlation $p_{\mcA}$ is
    extreme in $C_{qc}(X,X,\Z_2,\Z_2)$, and the assignment
    \begin{equation*}
        \tau\longmapsto(\tr_N\otimes\tau)\circ\lambda
    \end{equation*}
    gives a bijection between tracial states on $\mcA$ and tracial
    states on $\PVM{X,\Z_2}$ implementing $p_{\mcA}$. Under this
    correspondence, $\tau$ is finite-dimensional if and only if the
    corresponding implementing trace is finite-dimensional, and
    $\tau$ is Connes-embeddable if and only if the corresponding
    implementing trace is amenable.
\end{proposition}

We now apply this proposition to $\mcA_{\mathrm{rot}}$. Let
$\wtd{p}:=p_{\mcA_{\mathrm{rot}}}$ and define
$\wtd{\tau}_j:=(\tr_N\otimes\tau_j)\circ\lambda$ for
$j\in\{0,1\}$. The proposition gives
$\wtd{p}\in C_q^s(X,\Z_2)$ and shows that it is extreme, with
$\wtd{\tau}_0$ as its unique finite-dimensional implementing
tracial state. Since $\tau_1$ is amenable and hence
Connes-embeddable, $\wtd{\tau}_1$ is a distinct amenable
implementing tracial state. Therefore, by (a) and (b), $\wtd{p}$ is
a self-test but not a robust self-test, proving \Cref{thm:main}.

The unique finite-dimensional implementing trace also gives a natural
ideal strategy $\wtd{\mcS}$. Since $\tau_0$ is a character, the
GNS Hilbert space of $\wtd{\tau}_0$ is naturally identified with
$L^2(M_N(\C),\tr_N)\cong\C^N\otimes\C^N$, under which the
cyclic vector corresponds to the maximally entangled state
$\ket{\varphi_N}$. Thus $\wtd{p}$ self-tests a strategy employing
a maximally entangled state of Schmidt rank $N$, but this self-test
is not robust.

The correlation $\wtd{p}$ has binary answer set $\Z_2$, while the
size of its question set $X$ is far from optimized. We give an
explicit count in \Cref{sec:rotation-example}, with $\abs{X}=250$
and $N=36$. In the opposite direction, we show in \Cref{sec:lower}
that this phenomenon cannot occur in the smallest synchronous binary
scenarios: any synchronous binary correlation that is a self-test but
not a robust self-test must have at least four questions, and its
ideal state must have Schmidt rank at least three.

The first step toward proving
\Cref{thm:projection-construction-intro} is to encode the defining
relations using projections. For example, consider the relation
$rs-sr=0$, where $r,s$ are unitaries in a unital $C^*$-algebra
$\mcA$. Define the columns
\begin{equation*}
    \eta:=\frac{1}{\sqrt5}(0,\Id,s,rs,r,sr)^T,
    ~~~
    \xi:=\frac{1}{\sqrt3}(\Id,0,0,\Id,0,-\Id)^T
\end{equation*}
in $M_{6,1}(\mcA)$. Since $\eta^*\eta=\xi^*\xi=\Id$, the
matrices $E:=\eta\eta^*$ and $F:=\xi\xi^*$ are projections in
$M_6(\mcA)$. Moreover, $\xi^*\eta=(rs-sr)/\sqrt{15}$, so
$EF=0$ is equivalent to $rs=sr$. Both projections have trace $1/6$
under every matrix-amplified trace $\tr_6\otimes\tau$. When the
relation holds, their product has trace zero, so all products of two
projections in $\{E,F\}$ have tracial values independent of $\tau$.
For a general presentation, one column contains the required
monomials, and a scalar column for each relation selects its
coefficients. Orthogonality of the corresponding projections then
expresses all the defining relations.

The remaining difficulty is to obtain (P1) and (P3) while preserving
(P2). We add projections that determine matrix units and link the
column entries to products of a common family of unitary generators.
In the tracial GNS representation of any implementing trace, the
prescribed values determine the trace of the square of every real
linear combination of the projections and $\Id$. Faithfulness of
the induced trace therefore forces the same vanishing linear
combinations as in our construction. Likewise, a zero trace of a
product of two projections forces orthogonality. These identities
recover the matrix units, the unitaries, and the formulas for the
constructed projections. The orthogonality conditions then force
the original defining relations, and the universal property gives
the factorization in (P3).

Finally, we remark that properties (P1) and (P2) in
\Cref{thm:projection-construction-intro} hold for every finitely
generated unital $C^*$-algebra, without any finite-presentation or
real-coefficient assumption. The hypothesis on the presentation is
required for the proof of (P3). We expect this construction to have further applications connecting
tracial properties of $C^*$-algebras with nonlocal correlations and
device-independent certification. The correspondence in
\Cref{prop:trace-correspondence-intro} already yields the following
two applications, which we develop in \Cref{app:further-consequences}.

\begin{itemize}
    \item \textbf{Infinite-dimensional self-testing.}
    A modification of the rotation-algebra construction gives an
    algebra $\mcB$ with no finite-dimensional representations and a
    unique tracial state, which is amenable. The associated
    synchronous binary correlation $p_{\mcB}$ is extreme and has a
    unique implementing tracial state. This uniqueness makes
    $p_{\mcB}$ a commuting-operator self-test in the sense of
    \cite{PSZZ24}. At the same time, amenability implies that
    $p_{\mcB}$ is a limit of finite-dimensional quantum correlations,
    while the absence of finite-dimensional traces prevents an exact
    finite-dimensional realization. In fact, $p_{\mcB}$ has no
    spatial tensor-product realization, even in infinite
    dimensions~\cite{KPS18}. Thus, the construction gives another
    example of the non-closure of the quantum correlation
    set~\cite{Slofstra19,DPP19}, with a self-test as the limiting
    correlation. Other genuinely infinite-dimensional self-tests
    appear in concurrent work~\cite{PS26,KSZ,PSZ}.

    \item \textbf{Extremality without uniqueness.}
   Applying the construction to $\C\oplus\C$ gives an extreme
synchronous binary quantum correlation with two inequivalent
finite-dimensional realizations, one for each character of the algebra. This recovers the existence of extreme quantum correlations with inequivalent realizations established in \cite{TFRBK21}.
    Taking $\mcA=C(\mathbb T)$, the universal $C^*$-algebra
    generated by one unitary, gives a stronger example. Each
    evaluation character on $C(\mathbb T)$ gives a distinct
    finite-dimensional tracial state implementing the same
    correlation $p_{\mcA}$. The corresponding strategies all employ
    the same maximally entangled state of Schmidt rank $N$ and have
    local dimension $N$, but no two are related by a local dilation.
    Thus, a single extreme synchronous binary correlation can have
    uncountably many inequivalent realizations, even when the shared
    state and local dimensions are fixed.
\end{itemize}

\paragraph{\textbf{Acknowledgements}} R.C is supported by National Natural Science Foundation of China (No 62602562). Y.Z. is supported by the European Union via an ERC grant (QInteract, Grant No 101078107).
%\subsection*{Acknowledgments} 

\section{Preliminaries}

All Hilbert spaces are complex and separable. We use bra-ket notation, with inner products conjugate-linear in the first variable and linear in the second. Thus $\ket{\xi}\bra{\eta}$ denotes the operator sending $\ket{\zeta}$ to $\langle\eta|\zeta\rangle\ket{\xi}$. The norm of a vector and the operator norm are both denoted by $\norm{\cdot}$, with the meaning clear from the context.

\subsection{Algebras and states}

A \emph{$C^*$-algebra} is a complex Banach $*$-algebra $\mcA$ satisfying $\norm{a^*a}=\norm{a}^2$ for all $a\in\mcA$. All $C^*$-algebras and $*$-homomorphisms in this paper are assumed to be unital. We write $\Id$ for the identity, or $\Id_{\mcA}$ when the algebra needs to be specified. The algebra of bounded operators on a Hilbert space $\mcH$ is denoted by $\msB(\mcH)$, with identity $\Id_{\mcH}$. An element is \emph{positive}, written $a\geq0$, if it has the form $a=b^*b$. A \emph{projection} satisfies $P=P^*=P^2$, a \emph{unitary} satisfies $u^*u=uu^*=\Id$, and a \emph{partial isometry} is an element $v$ for which $v^*v$ and $vv^*$ are projections. For projections, $P\leq Q$ means $PQ=QP=P$. For a nonzero projection $q\in\mcA$, the \emph{corner} $q\mcA q$ is a $C^*$-algebra with identity $q$.

A \emph{$*$-representation} of $\mcA$ on $\mcH$ is a $*$-homomorphism $\pi:\mcA\arr\msB(\mcH)$. It is \emph{faithful} if it is injective, \emph{finite-dimensional} if $\mcH$ is finite-dimensional, and \emph{irreducible} if the only closed subspaces invariant under $\pi(\mcA)$ are $\{0\}$ and $\mcH$. The algebra $\mcA$ is \emph{residually finite-dimensional} if its finite-dimensional representations separate its elements. A family \emph{generates} $\mcA$ if the unital $*$-algebra it generates is norm-dense in $\mcA$.

A presentation
\begin{equation*}
  \mcA=C^*\ang{g_1,\ldots,g_d\mid\mcR}
\end{equation*}
means that $\mcA$ is generated by $g_1,\ldots,g_d$ subject to the relations $R=0$ for $R\in\mcR$, where $\mcR$ is a set of noncommutative $*$-polynomials in the formal generators $g_1,\ldots,g_d$ (so $g_1^*,\ldots,g_d^*$ are also allowed variables). It has the following universal property: every tuple $(b_1,\ldots,b_d)$ in a unital $C^*$-algebra $\mcB$ satisfying $R(b_1,\ldots,b_d)=0$ for all $R\in\mcR$ determines a unique $*$-homomorphism $\lambda:\mcA\arr\mcB$ with $\lambda(g_i)=b_i$. We sometimes write $R=0$ in place of $R$ when listing relations. The algebra $\mcA$ is \emph{finitely presented} if it admits such a presentation with finitely many generators and relations. This notation is used only when the universal algebra exists. Imposing additional equations means taking the quotient by the norm-closed two-sided $*$-ideal generated by those equations.

We write $M_n:=M_n(\C)=\msB(\C^n)$ and use the standard basis $\{\ket{i}:1\leq i\leq n\}$, with matrix units $\ket{i}\bra{j}$. The identity matrix is denoted by $\Id_n$. The ordinary and normalized matrix traces are
\begin{equation*}
  \operatorname{Tr}(a):=\sum_{i=1}^n a_{i,i}~\text{ and }~ \tr_n(a):=\frac1n\operatorname{Tr}(a).
\end{equation*}
The space of $n$-by-$m$ matrices over $\mcA$ is denoted by $M_{n,m}(\mcA)$, and we write $M_{n,m}:=M_{n,m}(\C)$ and $M_n(\mcA):=M_{n,n}(\mcA)$. In particular, elements of $M_{n,1}(\mcA)$ are called \emph{columns}, with the usual matrix multiplication and adjoint.

A \emph{state} on $\mcA$ is a linear functional $f:\mcA\arr\C$ such that $f(a^*a)\geq0$ for all $a\in\mcA$ and $f(\Id)=1$. It is \emph{faithful} if $f(a^*a)=0$ implies $a=0$. A state $\tau$ is \emph{tracial} if $\tau(ab)=\tau(ba)$ for all $a,b\in\mcA$, and $T(\mcA)$ denotes the set of tracial states. Weak-$*$ convergence of states means pointwise convergence on $\mcA$. The \emph{GNS representation} of $f$ consists of a representation $\pi_f:\mcA\arr\msB(\mcH_f)$ and a unit vector $\ket{f}$ such that
\begin{equation*}
  f(a)=\bra{f}\pi_f(a)\ket{f} ~\text{ and }~ \mcH_f=\overline{\{\pi_f(a)\ket{f}:a\in\mcA\}}.
\end{equation*}
The GNS representation $(\mcH_f,\pi_f,\ket{f})$ is unique up to unitary equivalence. Concretely, $\mcH_f$ is the completion of the vector-space quotient $\mcA/\{a:f(a^*a)=0\}$ for the inner product $\langle[a],[b]\rangle=f(a^*b)$, the representation acts by left multiplication, and $\ket{f}=[\Id]$.

For a tracial state $\tau$, the space $L^2(\mcA,\tau)$ is the completion of $\mcA$ in the seminorm $\norm{a}_{2,\tau}:=\tau(a^*a)^{1/2}$ after quotienting out its null space. Thus the GNS construction canonically identifies $L^2(\mcA,\tau)=\mcH_\tau$, with $[a]$ corresponding to $\pi_\tau(a)\ket{\tau}$.

A state $f$ is \emph{finite-dimensional} if its GNS Hilbert space $\mcH_f$ is finite-dimensional, equivalently, if it factors through a finite-dimensional $C^*$-algebra. We write $T_{\mathrm{fin}}(\mcA)$ for the convex set of finite-dimensional tracial states. A \emph{character} is a $*$-homomorphism $\chi:\mcA\arr\C$; it is a tracial state with one-dimensional GNS Hilbert space.

Recall that an element $s$ of a convex set $K$ is \emph{extreme} if $s=ts_1+(1-t)s_2$, with $0<t<1$ and $s_1,s_2\in K$, forces $s_1=s_2=s$. A convex subset $F\subseteq K$ is a \emph{face} if $ts_1+(1-t)s_2\in F$ under the same conditions forces $s_1,s_2\in F$.

For a tracial state $\tau$, the set $\{a:\tau(a^*a)=0\}$ is the closed two-sided ideal $\ker\pi_\tau$, and $\tau$ induces a faithful trace on the quotient $\mcA/\ker\pi_\tau\cong\pi_\tau(\mcA)$. We write
\begin{equation*}
  \mcM:=\pi_\tau(\mcA)''\subseteq\msB(\mcH_\tau)
\end{equation*}
for its double commutant, equivalently its weak-operator closure, which is a von Neumann algebra. The vector state $\tau_{\mcM}(x):=\bra{\tau}x\ket{\tau}$ is a faithful normal tracial state on $\mcM$ and satisfies $\tau=\tau_{\mcM}\circ\pi_\tau$. For projections $P,Q$ and a tracial state $\tau$, we will repeatedly use
\begin{equation*}
  \tau(PQ)=\tau(PQP)=\tau\bigl((QP)^*(QP)\bigr)\geq0.
\end{equation*}
Thus, when $\tau$ is faithful, $\tau(PQ)=0$ implies $PQ=0$; likewise, $P\leq Q$ and $\tau(P)=\tau(Q)$ imply $P=Q$.

The \emph{minimal tensor product} $\mcA\otimes_{\min}\mcB$ is the completion of the algebraic tensor product $\mcA\odot\mcB$ in the operator norm obtained from faithful representations of $\mcA$ and $\mcB$ on separate Hilbert spaces. This norm does not depend on the chosen faithful representations.

We always identify
\begin{equation*}
  M_{n,m}(\mcA)=M_{n,m}\otimes\mcA,
\end{equation*}
where the tensor product here is algebraic: representing $\mcA$ faithfully on a Hilbert space $\mcH$, we view these matrices as operators from $\C^m\otimes\mcH$ to $\C^n\otimes\mcH$, with the corresponding operator norm. Thus we use
\begin{equation*}
  (a_{i,j})_{1\leq i\leq n,\,1\leq j\leq m} \qquad\text{and}\qquad \sum_{i=1}^n\sum_{j=1}^m \ket{i}\bra{j}\otimes a_{i,j}
\end{equation*}
interchangeably for elements of $M_{n,m}(\mcA)$, where each $a_{i,j}\in\mcA$. In the square case, this identification realizes the \emph{matrix amplification} $M_n(\mcA)$ as the $C^*$-algebra $M_n\otimes_{\min}\mcA$.\footnote{Since $M_n$ is nuclear, we also have $M_n\otimes_{\min}\mcA=M_n\otimes_{\max}\mcA$.} A $*$-homomorphism $\lambda:\mcA\arr\mcB$ amplifies entrywise to $\operatorname{id}_{M_n}\otimes\lambda:M_n(\mcA)\arr M_n(\mcB)$. For $\tau\in T(\mcA)$, its matrix amplification on $M_n(\mcA)$ is the tracial state
\begin{equation*}
  (\tr_n\otimes\tau)\bigl((a_{i,j})_{i,j=1}^n\bigr) :=\frac1n\sum_{i=1}^n\tau(a_{i,i}).
\end{equation*}
Every tracial state $\sigma$ on $M_n(\mcA)$ has this form for a unique $\tau$, recovered by $\tau(a)=\sigma(\Id_n\otimes a)$.

The \emph{opposite algebra} $\mcA^{\mathrm{op}}$ has the same underlying vector space, involution, and norm as $\mcA$, with multiplication reversed: $a^{\mathrm{op}}b^{\mathrm{op}}=(ba)^{\mathrm{op}}$. A tracial state $\tau$ on $\mcA$ is \emph{amenable} if the linear functional on $\mcA\odot\mcA^{\mathrm{op}}$ given by
\begin{equation*}
  a\otimes b^{\mathrm{op}}\longmapsto\tau(ab)
\end{equation*}
extends to a state on $\mcA\otimes_{\min}\mcA^{\mathrm{op}}$; see \cite[Theorem~3.1.6]{Brown06}. A $C^*$-algebra $\mcA$ is \emph{nuclear} if $\mcA\odot\mcB$ has a unique $C^*$-norm for every $C^*$-algebra $\mcB$. Every tracial state on a nuclear $C^*$-algebra is amenable~\cite[Section~4.2]{Brown06}.

The following lemma records some properties of tracial states that we need.

\begin{lemma}\label{lemma:pre_traces}
  Let $\mcA$ and $\mcB$ be two $C^*$-algebras, and let $\lambda:\mcA\arr M_N(\mcB)$ be a surjective homomorphism for some $N\geq 1$. For any tracial state $\tau$ on $\mcB$, let $\wtd{\tau}:=(\tr_N\otimes \tau)\circ \lambda$ be a tracial state on $\mcA$. Then:
  \begin{enumerate}[(a)]
    \item The tracial state $\tau$ is finite-dimensional if and only if the corresponding $\wtd{\tau}$ if finite-dimensional.
    \item If $\tau$ is amenable, then the corresponding $\wtd{\tau}$ is amenable.
  \end{enumerate}
\end{lemma}
\begin{proof}
  Since $\lambda$ is surjective, pullback along $\lambda$ preserves the GNS Hilbert space. Moreover, the GNS space of the matrix-amplified trace satisfies
  \begin{equation*}
    \mcH_{\wtd{\tau}} \cong \mcH_{\tr_N\otimes\tau} \cong L^2(M_N(\C),\tr_N)\otimes\mcH_\tau \cong \C^{N^2}\otimes\mcH_\tau.
  \end{equation*}
  This proves part~(a).

  Part~(b) follows from the standard facts that amenability of tracial states is preserved under matrix amplification and under pullback by a $*$-homomorphism.
\end{proof}

Note that parts (a) and (b) have different functorial behaviour. Surjectivity of $\lambda$ is what makes the GNS space of the pullback $\wtd{\tau}$ canonically equivalent to that of $\tr_N\otimes\tau$, and hence gives the equivalence in part~(a). By contrast, surjectivity is not needed for part~(b): amenability is preserved under pullback by any unital $*$-homomorphism. The converse of part~(b) is false in general, because amenability need not descend through quotient maps~\cite{Brown06}.

\begin{lemma}[{\cite[Lemma~5.4(b) and Proposition~5.6]{Zhao24}}]\label{lem:matrix-trace-extreme}
  Let $\mcA$ be a $C^*$-algebra and let
  \begin{equation*}
    \pi:\mcA\longrightarrow M_d(\C)
  \end{equation*}
  be a surjective $*$-homomorphism. Then the tracial state
  \begin{equation*}
    \tau:=\tr_d\circ\pi
  \end{equation*}
  is an extreme point of $T_{\mathrm{fin}}(\mcA)$. Moreover, the GNS representation of $\tau$ is unitarily equivalent to
  \begin{equation*}
    \big( \C^d\otimes\C^d, \pi(\,\cdot\,)\otimes\Id, \ket{\varphi_d} \big),
  \end{equation*}
  where $\ket{\varphi_d} := \frac{1}{\sqrt d} \sum_{i=1}^d \ket{i}\otimes\ket{i}$ is a maximally entangled state in $\C^d\otimes\C^d$.
\end{lemma}

\subsection{Correlations and strategies}
A (bipartite) \emph{Bell scenario} $(X,Y,A,B)$ consists of question sets $X,Y$ and answer sets $A,B$ for Alice and Bob, respectively. All question and answer sets are finite and nonempty.

A \emph{positive operator-valued measure} (POVM) on a Hilbert space $\mcH$ with outcome set $A$ is a family of positive operators $\{P_a:a\in A\}\subseteq\msB(\mcH)$ satisfying $\sum_{a\in A}p_{\mcA}=\Id_{\mcH}$. It is a \emph{projection-valued measure} (PVM) if every $p_{\mcA}$ is a projection; these projections are then pairwise orthogonal. A \emph{correlation} in the scenario $(X,Y,A,B)$ is an array of nonnegative real numbers $p(a,b|x,y)$ satisfying $\sum_{a\in A,b\in B}p(a,b|x,y)=1$ for every $x\in X,y\in Y$. A correlation is \emph{binary} if $A=B=\Z_2$. We always use representatives $\{0,1\}$ for $\Z_2$.

\begin{definition}\label{def:strategy}
  A \emph{quantum strategy} in the scenario $(X,Y,A,B)$ is a tuple
  \begin{equation*}
    S=\bigl( \ket{\psi}\in\mcH_A\otimes\mcH_B, \{P_{x,a}\},\{Q_{y,b}\} \bigr),
  \end{equation*}
  where $\mcH_A,\mcH_B$ are finite-dimensional Hilbert spaces, $\ket{\psi}$ is a unit vector, and $\{P_{x,a}:a\in A\}$ and $\{Q_{y,b}:b\in B\}$ are POVMs on $\mcH_A$ and $\mcH_B$, respectively, for every $x\in X$ and $y\in Y$. Its \emph{quantum correlation} is
  \begin{equation*}
    p_S(a,b|x,y) :=\bra{\psi}P_{x,a}\otimes Q_{y,b}\ket{\psi}.
  \end{equation*}
\end{definition}

A strategy is \emph{projective} if all its POVMs are PVMs. For its shared state, the reduced density operators are
\begin{equation*}
  \rho_A:=\operatorname{Tr}_B(\ket{\psi}\bra{\psi})~\text{ and }~\rho_B:=\operatorname{Tr}_A(\ket{\psi}\bra{\psi}),
\end{equation*}
where $\operatorname{Tr}_A$ and $\operatorname{Tr}_B$ denote partial traces. The \emph{Schmidt rank} of $\ket{\psi}$ is $\operatorname{rank}\rho_A=\operatorname{rank}\rho_B$. We call the state, and the strategy, \emph{full-rank} if both reduced density operators are invertible. In particular, both local spaces then have dimension equal to the Schmidt rank. For local spaces of dimension $d$, the state is \emph{maximally entangled} if $\rho_A=\tfrac{1}{d}\Id_{\mcH_A}$ and $\rho_B=\tfrac{1}{d}\Id_{\mcH_B}$, as for $\ket{\varphi_d}$ above.

Unless explicitly stated otherwise, strategies have pure shared states as in \Cref{def:strategy}. When mixed states are allowed, the shared state is instead a density operator $\rho\geq0$ on $\mcH_A\otimes\mcH_B$ with $\operatorname{Tr}(\rho)=1$, and the correlation is given by $p_S(a,b|x,y)=\operatorname{Tr}\bigl(\rho(P_{x,a}\otimes Q_{y,b})\bigr)$. All local Hilbert spaces remain finite-dimensional.

We denote the set of quantum correlations by $C_q(X,Y,A,B)$. It is convex, as can be seen by taking direct sums of strategies. An \emph{extreme quantum correlation} is an extreme point of this set. In general, $C_q(X,Y,A,B)$ is not closed in the Euclidean topology~\cite{Slofstra19}. Its closure is denoted by $C_{qa}(X,Y,A,B)$, the set of \emph{quantum-approximate correlations}. Every quantum correlation is nonsignalling, meaning that the marginal probabilities
\begin{equation*}
  p_{\mcA}(a|x):=\sum_{b\in B}p(a,b|x,y)~\text{ and }~ p_B(b|y):=\sum_{a\in A}p(a,b|x,y)
\end{equation*}
are independent of $y$ and $x$, respectively.

In a scenario $(X,X,A,A)$, a correlation is \emph{synchronous} if
\begin{equation*}
  p(a,b|x,x)=0 \qquad\text{for all }x\in X\text{ and distinct }a,b\in A.
\end{equation*}
We write $C_q^s(X,A)$ and $C_{qa}^s(X,A)$ for the synchronous correlations in $C_q(X,X,A,A)$ and $C_{qa}(X,X,A,A)$, respectively. Nonnegativity of the probabilities shows that $C_q^s(X,A)$ is a face of $C_q(X,X,A,A)$: a convex combination can vanish on all forbidden synchronous events only if each summand does. Thus a synchronous quantum correlation is extreme in $C_q^s(X,A)$ if and only if it is extreme in $C_q(X,X,A,A)$. We call it an \emph{extreme synchronous quantum correlation}.

A correlation is \emph{deterministic} if there are functions $f:X\arr A$ and $g:Y\arr B$ such that $p(a,b|x,y)=\delta_{a,f(x)}\delta_{b,g(y)}$, and \emph{local} if it is a convex combination of deterministic correlations. Since deterministic correlations are quantum, an extreme quantum correlation that is local must be deterministic.

\section{Self-testing, robustness, and operator-algebraic formulations}\label{sec:self-testing}
In this section, we make precise the notions of self-testing and robust self-testing used in this paper. Given two correlations $p,q\in C_q(X,Y,A,B)$, we define
\begin{equation*}
  \norm{p-q}_1 := \sum_{\substack{x\in X,\ y\in Y\\a\in A,\ b\in B}} \abs{p(a,b|x,y)-q(a,b|x,y)}.
\end{equation*}

\begin{definition}[Local $\epsilon$-dilation]
  \label{def:local-epsilon-dilation}
  Let
  \begin{equation*}
    S = \big( \ket{\psi}\in\mcH_A\otimes\mcH_B, \{P_{x,a}\}, \{Q_{y,b}\} \big)
  \end{equation*}
  and
  \begin{equation*}
    \wtd{S} = \big( \ket{\wtd{\psi}}\in \wtd{\mcH}_A\otimes\wtd{\mcH}_B, \{\wtd{P}_{x,a}\}, \{\wtd{Q}_{y,b}\} \big)
  \end{equation*}
  be two quantum strategies, and let $\epsilon\geq 0$. We say that $\wtd{S}$ is a \emph{local $\epsilon$-dilation} of $S$ if there exist finite-dimensional Hilbert spaces $\mcK_A$ and $\mcK_B$, isometries
  \begin{equation*}
    V_A:\mcH_A\longrightarrow\wtd{\mcH}_A\otimes\mcK_A, \qquad V_B:\mcH_B\longrightarrow\wtd{\mcH}_B\otimes\mcK_B,
  \end{equation*}
  and a unit vector
  \begin{equation*}
    \ket{\kappa}\in\mcK_A\otimes\mcK_B
  \end{equation*}
  such that
  \begin{equation*}
    \norm{ (V_A\otimes V_B) (P_{x,a}\otimes Q_{y,b})\ket{\psi} - \big( (\wtd{P}_{x,a}\otimes\wtd{Q}_{y,b}) \ket{\wtd{\psi}} \big) \otimes\ket{\kappa} } \leq \epsilon
  \end{equation*}
  for all $x\in X$, $y\in Y$, $a\in A$, and $b\in B$. When $\epsilon=0$, we simply say that $\wtd{S}$ is a \emph{local dilation} of $S$.
\end{definition}

In \Cref{def:local-epsilon-dilation}, we implicitly identify
\begin{equation*}
  (\wtd{\mcH}_A\otimes\mcK_A) \otimes (\wtd{\mcH}_B\otimes\mcK_B)
\end{equation*}
with
\begin{equation*}
  (\wtd{\mcH}_A\otimes\wtd{\mcH}_B) \otimes (\mcK_A\otimes\mcK_B)
\end{equation*}
by exchanging the two middle tensor factors.

\begin{definition}[Self-testing and robust self-testing]
  \label{def:self-testing}
  Let $\mcC$ be a class of quantum strategies, let $\wtd{S}\in\mcC$, and let $\wtd{p}:=p_{\wtd{S}}$.

  We say that $\wtd{p}$ \emph{self-tests} $\wtd{S}$ for $\mcC$ if $\wtd{S}$ is a local dilation of every $S\in\mcC$ satisfying
  \begin{equation*}
    p_{S}=\wtd{p}.
  \end{equation*}

  We say that $\wtd{p}$ \emph{robustly self-tests} $\wtd{S}$ for $\mcC$ if there is a function $\delta:\R_{\geq 0}\longrightarrow\R_{\geq 0}$ with $\delta(\epsilon)\to 0$ as $\epsilon\to 0$ such that, for every $\epsilon\geq 0$ and every $S\in\mcC$ satisfying
  \begin{equation*}
    \norm{p_{S}-\wtd{p}}_1\leq\epsilon,
  \end{equation*}
  the strategy $\wtd{S}$ is a local $\delta(\epsilon)$-dilation of $S$.
\end{definition}

The strategy $\wtd{S}$ in \Cref{def:self-testing} is commonly referred to as the \emph{ideal strategy}, and strategies in $\mcC$ are called \emph{employed strategies}. When the ideal strategy is clear from the context, we simply say that $\wtd{p}$ is a self-test or a robust self-test for $\mcC$. The direction in the term ``local dilation'' is important: the ideal strategy $\wtd{S}$ is extracted from the employed strategy $S$, possibly together with an auxiliary system.

\subsection{Lifting assumptions on the employed strategies}\label{subsec:lifting} The definition of self-testing depends a priori on the class $\mcC$ of employed strategies. Suppose $\mcC_1$ and $\mcC_2$ are two classes of quantum strategies with $\mcC_1\subsetneq \mcC_2$. If a correlation self-tests an ideal strategy $\wtd{S}\in\mcC_1$ for the larger class $\mcC_2$, then it automatically self-tests the same strategy $\wtd{S}$ for the more restricted class $\mcC_1$. The converse, however, does not hold in general.

In many cases, restrictions on the employed strategies make a self-testing result considerably easier to prove. In particular, we sometimes assume all employed strategies use projective measurements, or share a full-rank quantum state, resulting in the class $\mcC_{\mathrm{PVM}}$ of PVM strategies and the class $\mcC_{\mathrm{full}}$ of full-rank strategies, respectively. Quantum strategies with arbitrary POVM measurements, as defined in \Cref{def:strategy}, form the class $\mcC_{\mathrm{POVM}}$. Strategies in $\mcC_{\mathrm{POVM}}$ are still assumed to share a pure quantum state. If arbitrary mixed states are also allowed, we obtain the class $\mcC_{\mathrm{free}}$ of so-called assumption-free strategies.

Since self-testing for a restricted class does not necessarily imply self-testing for a larger class, it is natural to ask when assumptions on the employed strategies can be lifted. This question was first studied in \cite{PSZZ24}, where it was shown that, for synchronous correlations and for extreme binary correlations, self-testing for $\mcC_{\mathrm{PVM}}$ implies self-testing for $\mcC_{\mathrm{POVM}}$.

The problem was subsequently studied systematically in \cite{BCKLNS25}, for both self-testing and robust self-testing. Their main results can be summarized as follows. Let $\wtd{p}$ be an \emph{extreme} quantum correlation and let $\wtd{S}\in\mcC_{\mathrm{POVM}}$ be an ideal strategy for $\wtd{p}$. If $\wtd{S}$ is full-rank, then robust self-testing for $\mcC_{\mathrm{PVM}}$ implies robust self-testing for $\mcC_{\mathrm{free}}$. Similarly, if $\wtd{S}$ is projective, then robust self-testing for $\mcC_{\mathrm{full}}$ implies robust self-testing for $\mcC_{\mathrm{free}}$. The same implications hold for exact self-testing. Here, extremality is needed in the final step that lifts the purity assumption and allows arbitrary mixed employed states. In particular, when the ideal strategy is both full-rank and projective, a robust self-testing result proved under either the PVM assumption or the full-rank assumption can be upgraded to an assumption-free one.

Synchronous correlations allow an even stronger reduction. Every extreme synchronous quantum correlation admits a \emph{projective maximally entangled strategy} \cite{PSSTW16}, that is, a strategy employing a maximally entangled state and projective measurements. We denote the class of such strategies by $\mcC_{\mathrm{PME}}$. We see that
\begin{equation*}
  \mcC_{\mathrm{PME}} \subseteq \mcC_{\mathrm{PVM}}\cap\mcC_{\mathrm{full}}\subseteq \mcC_{\mathrm{POVM}}\subseteq \mcC_{\mathrm{free}}.
\end{equation*}

The lifting theorem of \cite{VZ25} shows that, for extreme synchronous quantum correlations, robust self-testing for $\mcC_{\mathrm{PME}}$ can be lifted to robust self-testing for $\mcC_{\mathrm{POVM}}$, with a quantitative relation between the two robustness functions. Combining this result with the $\mcC_{\mathrm{POVM}}$-to-$\mcC_{\mathrm{free}}$ lifting from \cite{BCKLNS25}, we conclude that:
\begin{theorem}
  Let $p\in C_q^s(X,A)$ be an extreme synchronous quantum correlation. Then $p$ is a robust self-test for $\mcC_{\mathrm{PME}}$ if and only if $p$ is a robust self-test for $\mcC_{\mathrm{free}}$. The corresponding statement also holds for exact self-testing.
\end{theorem}

It follows that, for extreme synchronous quantum correlations, all the classes of strategies considered above give the same notions of self-testing and robust self-testing. We will therefore omit the class of employed strategies and simply say that an extreme synchronous correlation is a self-test or a robust self-test.

\subsection{Tracial-states characterizations}
Quantum correlations can be naturally described by states on $C^*$-algebras. Let $\POVM{X,A}$ denote the universal $C^*$-algebra generated by positive elements
\begin{equation*}
  \{e_{x,a}:x\in X,\ a\in A\}
\end{equation*}
satisfying
\begin{equation*}
  \sum_{a\in A}e_{x,a}=\Id
\end{equation*}
for every $x\in X$. Thus, for each $x\in X$, the generators $\{e_{x,a}:a\in A\}$ form an ``abstract'' POVM. We define $\PVM{X,A}$ to be the quotient of $\POVM{X,A}$ by the additional relations
\begin{equation*}
  e_{x,a}^2=e_{x,a}
\end{equation*}
for all $x\in X,a\in A.$ Hence, $\PVM{X,A}$ is the universal $C^*$-algebra generated by ``abstract'' PVMs.

A state $f$ on $\POVM{X,A}\otimes_{\min}\POVM{Y,B}$ is said to \emph{implement} a correlation $p$ if
\begin{equation*}
  f(e_{x,a}\otimes e_{y,b}) = p(a,b|x,y)
\end{equation*}
for all $x\in X$, $y\in Y$, $a\in A$, and $b\in B$. We use the same terminology for states on $\PVM{X,A}\otimes_{\min}\PVM{Y,B}$.

By the universal property, every quantum strategy $S$ induces representations $\pi_A$ and $\pi_B$ of the corresponding POVM algebras, and hence a finite-dimensional state
\begin{equation*}
  f_{S}(\alpha\otimes\beta) := \bra{\psi} \pi_A(\alpha)\otimes\pi_B(\beta) \ket{\psi}
\end{equation*}
implementing the correlation $p_{S}$. If $S$ is projective, then $f_{S}$ factors through the PVM algebra $\PVM{X,A}\otimes_{\min}\PVM{Y,B}$. Conversely, by the GNS construction, every finite-dimensional state on either tensor-product algebra arises from a quantum strategy of the corresponding type.

The following result from \cite{PSZZ24} relates uniqueness of implementing states to self-testing.

\begin{theorem}\label{thm:state-self-testing}
  Let $p\in C_q(X,Y,A,B)$ be an extreme quantum correlation.
  \begin{enumerate}[(a)]
    \item The correlation $p$ is a self-test for $\mcC_{\mathrm{POVM}}$ if and only if there is a unique finite-dimensional state on
          \begin{equation*}
            \POVM{X,A}\otimes_{\min}\POVM{Y,B}
          \end{equation*}
          implementing $p$.

    \item Suppose, in addition, that $p$ admits a full-rank projective strategy. Then $p$ is a self-test for $\mcC_{\mathrm{PVM}}$ if and only if there is a unique finite-dimensional state on
          \begin{equation*}
            \PVM{X,A}\otimes_{\min}\PVM{Y,B}
          \end{equation*}
          implementing $p$.
  \end{enumerate}
\end{theorem}

It is well known that $\POVM{X,A}\otimes_{\min}\POVM{Y,B}$ is residually finite-dimensional (see \cite[Appendix~A]{Zhao24} for an elementary proof), so finite-dimensional states are weak-$*$ dense in the state space. This argument connects uniqueness among all implementing states with robust self-testing.

\begin{theorem}[{\cite{Zhao24,Kar25}}]
  \label{thm:state-robust-self-testing}
  Let $p\in C_q(X,Y,A,B)$ be a quantum correlation.
  \begin{enumerate}[(a)]
    \item If $p$ is a robust self-test for $\mcC_{\mathrm{POVM}}$, then there is a unique state on
          \begin{equation*}
            \POVM{X,A}\otimes_{\min}\POVM{Y,B}
          \end{equation*}
          implementing $p$.

    \item Suppose $p$ is an extreme synchronous quantum correlation. Then $p$ is a robust self-test if and only if there is a unique state on
          \begin{equation*}
            \PVM{X,A}\otimes_{\min}\PVM{X,A}
          \end{equation*}
          implementing $p$.
  \end{enumerate}
\end{theorem}

Even for extreme quantum correlations, it is not known whether the converse of \Cref{thm:state-robust-self-testing}(a) holds in general. Besides synchronous correlations, the converse is known for several other structured classes, including unique optimal correlations of XOR games.

\begin{proposition}[{\cite[Theorem~7.14]{Zhao24}}]
  \label{thm:unbiased-even-rank}
  Let $p\in C_q^s(X,\Z_2)$ be an extreme synchronous quantum correlation with $p_{\mcA}(1|x)=1/2$ for every $x\in X$. Define
  \begin{equation*}
    C_{x,y}:= \sum_{a,b\in\Z_2}(-1)^{a+b}p(a,b|x,y) =4p(1,1|x,y)-1.
  \end{equation*}
  If $\operatorname{rank}C$ is even, then $p$ is a robust self-test.
\end{proposition}

For synchronous correlations, the characterization in terms of states on the tensor product of two algebras reduces to one involving tracial states on a single algebra. We say that a tracial state $\tau$ on $\PVM{X,A}$ \emph{implements} a synchronous correlation $p$ if
\begin{equation*}
  \tau(e_{x,a}e_{y,b}) = p(a,b|x,y)
\end{equation*}
for all $x,y\in X$ and $a,b\in A$. The following correspondence, first established in \cite{PSSTW16}, demonstrates that synchronous quantum correlations are characterized by finite-dimensional implementing tracial states.
\begin{theorem}\label{thm:synchronous-quantum-fd-trace}
  A correlation $p$ is in $C_q^s(X,A)$ if and only if there is a finite-dimensional tracial state on $\PVM{X,A}$ implementing $p$.
\end{theorem}

Let $\omega$ be the canonical $*$-anti-automorphism of $\PVM{X,A}$ fixing every generator $e_{x,a}$, so that $\omega$ reverses the order of products.\footnote{Equivalently, $\alpha\mapsto\omega(\alpha)^{\mathrm{op}}$ is the canonical $*$-isomorphism from $\PVM{X,A}$ onto its opposite algebra $\big(\PVM{X,A}\big)^{\mathrm{op}}$.} A tracial state $\tau$ on $\PVM{X,A}$ is amenable if and only if the formula
\begin{equation}
  f_\tau(\alpha\otimes\beta) := \tau\big(\alpha\omega(\beta)\big)
  \label{eq:trace-associated-state}
\end{equation}
extends to a state on $\PVM{X,A}\otimes_{\min}\PVM{X,A}$. In this case, $\tau$ and $f_\tau$ implement the same synchronous correlation. Conversely, every state on $\PVM{X,A}\otimes_{\min}\PVM{X,A}$ implementing a synchronous correlation arises uniquely in this way from an amenable tracial state on $\PVM{X,A}$. Moreover, under this correspondence, $\tau$ is finite-dimensional if and only if the corresponding $f_\tau$ is finite-dimensional.

Combining this correspondence with \Cref{thm:state-self-testing,thm:state-robust-self-testing} gives the following characterization.

\begin{theorem}[{\cite{Zhao24,Kar25}}]
  \label{thm:tracial-self-testing}
  Let $p\in C_q^s(X,A)$ be an extreme synchronous quantum correlation.
  \begin{enumerate}[(a)]
    \item The correlation $p$ is a self-test if and only if there is a unique finite-dimensional tracial state on $\PVM{X,A}$ implementing $p$.

    \item The correlation $p$ is a robust self-test if and only if there is a unique amenable tracial state on $\PVM{X,A}$ implementing $p$.
  \end{enumerate}
\end{theorem}

Thus, for extreme synchronous quantum correlations, the difference between exact and robust self-testing is precisely the difference between uniqueness among finite-dimensional tracial states and uniqueness among amenable tracial states. This is the formulation that we use throughout the remainder of the paper.

The unique finite-dimensional implementing tracial state in \Cref{thm:tracial-self-testing}(a) also determines a natural ideal strategy, up to local unitary equivalence.

\begin{theorem}[{\cite[Corollary 6.7]{Zhao24}}]\label{thm:ideal-strategy-from-trace}
  Let $p\in C_q^s(X,A)$ be an extreme synchronous quantum correlation that is a self-test, and let $\tau$ be the unique finite-dimensional tracial state on $\PVM{X,A}$ implementing $p$. Then there exist $d\geq 1$ and an irreducible $*$-representation
  \begin{equation*}
    \pi:\PVM{X,A}\longrightarrow M_d(\C)
  \end{equation*}
  such that
  \begin{equation*}
    \tau=\tr_d\circ\pi,
  \end{equation*}
  The GNS representation of $\tau$ is unitarily equivalent to
  \begin{equation*}
    \big( \C^d\otimes\C^d, \pi(\,\cdot\,)\otimes\Id, \ket{\varphi_d} \big),
  \end{equation*}
  where
  \begin{equation*}
    \ket{\varphi_d} := \frac{1}{\sqrt d} \sum_{i=1}^d \ket{i}\otimes\ket{i}
  \end{equation*}
  is a maximally entangled state in $\C^d\otimes\C^d$. Moreover, $p$ self-tests the ideal strategy
  \begin{equation*}
    \wtd{S} := \big( \ket{\varphi_d}, \{P_{x,a}\}, \{Q_{y,b}\} \big),
  \end{equation*}
  where
  \begin{equation*}
    P_{x,a}:=\pi(e_{x,a}), \qquad Q_{y,b}:=\pi(e_{y,b})^{T},
  \end{equation*}
  and the transpose is taken with respect to the basis $\{\ket{i}:1\leq i\leq d\}$.
\end{theorem}

Finally, we recall the following results from \cite{Zhao24}: if a synchronous correlation $p$ admits a unique \emph{extreme} finite-dimensional implementing tracial state, then $p$ itself must be extreme.

\begin{lemma}\label{lem:unique-extreme-trace-implies-extreme-correlation}
  Suppose that $p\in C_q^s(X,A)$ has a unique finite-dimensional implementing tracial state $\tau$. If $\tau$ is an extreme point of $T_{\mathrm{fin}}(\PVM{X,A})$, then $p$ is an extreme synchronous quantum correlation.
\end{lemma}
Recall that $T_{\mathrm{fin}}(\PVM{X,A})$ denotes the set of finite-dimensional tracial states on $\PVM{X,A}$. For completeness, we include a short proof.
\begin{proof}
  Suppose $p=t p_1+(1-t)p_2$ for some $0<t<1$ and $p_1,p_2\in C_q(X,X,A,A)$. Since $C_q^s(X,A)$ is a face of $C_q(X,X,A,A)$, both $p_1,p_2\in C_q^s(X,A)$. By \Cref{thm:synchronous-quantum-fd-trace}, there exist finite-dimensional tracial states $\tau_1$ and $\tau_2$ on $\PVM{X,A}$ implementing $p_1$ and $p_2$, respectively. Then $ t\tau_1+(1-t)\tau_2$ is a finite-dimensional tracial state implementing $p$, and hence equals $\tau$ by uniqueness. Since $\tau$ is extreme, $\tau_1=\tau_2=\tau$. Consequently, $p_1=p_2=p$. This proves that $p$ is extreme in $C_q(X,X,A,A)$, and hence extreme in $C_q^s(X,A)$.
\end{proof}

\section{Encoding algebras by projections}
\label{sec:projections_generating}

\begin{definition}
  \label{def:real-coefficient-unitary}
  A unital $C^*$-algebra $\mcA$ has a \emph{finite real-coefficient unitary presentation} if there are integers $d,m\geq1$, unitaries $u_1,\ldots,u_d\in\mcA$, a finite set $\mcS$ of formal $*$-monomials in $u_1,\ldots,u_d$, and real coefficients $\lambda_{j,w}$, $w\in\mcS$, $1\leq j\leq m$, such that $\mcA$ is the universal unital $C^*$-algebra generated by $u_1,\ldots,u_d$ subject to the relations
  \begin{equation}
    R_j:=\sum_{w\in\mcS}\lambda_{j,w}w=0,\quad 1\leq j\leq m.
    \label{eq:pc-presentation}
  \end{equation}
\end{definition}
Here a $*$-monomial in $u_1,\ldots,u_d$ is a finite product of the variables $u_1,\ldots,u_d,u_1^*,\ldots,u_d^*$, with the empty product allowed and denoted by $\Id$. We fix a presentation of this form, collect identical formal $*$-monomials, and use coefficient zero when a monomial does not appear in a particular relation. We may assume that every $w\in\mcS$ has a nonzero coefficient in at least one relation.

%When a $*$-monomial is used to label a matrix index, it refers to itsformal expression. When it appears as a matrix entry, it denotes its value in $\mcA$. For instance, $u_1^*u_1$ and $\Id$ label different indices, although their values in $\mcA$ are equal. For another tuple of unitaries $U=(U_1,\ldots,U_d)$, write $w(U)$ for the product obtained by replacing each $u_i$ in $w$ with $U_i$ and each $u_i^*$ with $U_i^*$, and set $\Id(U):=\Id$.

The purpose of this section is to prove the following construction.
\begin{theorem}
  \label{thm:projection-construction}
  Suppose $\mcA$ has a finite real-coefficient unitary presentation and admits a tracial state. There exist a positive integer $N$, a finite set $X$, and projections $P_x\in M_N(\mcA)$, $x\in X$, with the following properties.
  \begin{enumerate}[(a)]
    \item The projections $P_x$ generate $M_N(\mcA)$. In particular, the unital $*$-homomorphism
          \begin{equation*}
            \Phi:\PVM{X,\Z_2}\arr M_N(\mcA) ~\text{ sending }~ e_{x,1}\mapsto P_x
          \end{equation*}
          is surjective.

    \item There is a real array $(p_{x,y})_{x,y\in X}$ such that
          \begin{equation*}
            (\tr_N\otimes\tau)(P_xP_y)=p_{x,y}
          \end{equation*}
          for every tracial state $\tau$ on $\mcA$ and all $x,y\in X$.

    \item If a tracial state $\varphi$ on $\PVM{X,\Z_2}$ satisfies $\varphi(e_{x,1}e_{y,1})=p_{x,y}$ for all $x,y\in X$, then there is a unique tracial state $\tau_\varphi$ on $\mcA$ such that
          \begin{equation*}
            \varphi=(\tr_N\otimes\tau_\varphi)\circ\Phi.
          \end{equation*}
  \end{enumerate}
\end{theorem}

Our starting point is to encode each defining relation using columns in $M_{N,1}(\mcA)$ and their range projections in $M_N(\mcA)$. For example, consider the commutation relation $xy-yx=0$, where $x,y$ are unitaries in a unital $C^*$-algebra $\mcA$. Define
\begin{equation*}
  \eta:=\frac1{\sqrt5} \begin{pmatrix}0&\Id&y&xy&x&yx\end{pmatrix}^{T} ~\text{ and }~ \xi:=\frac1{\sqrt3} \begin{pmatrix}\Id&0&0&\Id&0&-\Id\end{pmatrix}^{T}
\end{equation*}
in $M_{6,1}(\mcA)$. Since $\eta^*\eta=\xi^*\xi=\Id$, the matrices $E:=\eta\eta^*$ and $F:=\xi\xi^*$ are projections in $M_6(\mcA)$. Moreover,
\begin{equation*}
  \xi^*\eta=\frac{xy-yx}{\sqrt{15}},
\end{equation*}
so $xy-yx=0$ if and only if $\xi^*\eta=0$, or equivalently $EF=0$. We explain how these columns are chosen from the relation in \Cref{ex:pc-commutation}.

When $xy-yx=0$, every tracial state $\tau$ on $\mcA$ satisfies
\begin{equation*}
  (\tr_6\otimes\tau)(E^2) =(\tr_6\otimes\tau)(F^2)=\frac16 ~\text{ and }~ (\tr_6\otimes\tau)(EF)=0.
\end{equation*}
Thus the traces of all products of two projections in $\{E,F\}$ are independent of $\tau$, as required in part~(b). In the general construction, we will use one column $\eta$ containing all the required $*$-monomials and a separate scalar column $\xi_i$ for each relation $R_i$. Their projections $E,F_1,\ldots,F_m$ will express the defining relations through $EF_j=0$ for all $1\leq j\leq m$.

These projections form the starting point of the generating family. The main difficulty is not to add enough projections to generate $M_N(\mcA)$, but to do so while preserving part~(b) and ensuring part~(c). In particular, the prescribed traces of products must force every corresponding tracial GNS representation to factor through $\Phi$. For this, we first establish some technical lemmas that construct the additional projections and determine their form from traces of products.

\subsection{Elementary projection formulas}

A column $\zeta\in M_{N,1}(\mcA)$ is \emph{normalized} if $\zeta^*\zeta=\Id_{\mcA}$. We call a projection $Q\in M_N(\mcA)$ \emph{rank one over $\mcA$} if $Q=\zeta_Q\zeta_Q^*$ for a normalized column $\zeta_Q\in M_{N,1}(\mcA)$. Whenever we use this notation, a column has been chosen, and it need not be unique. This does not assert that $Q$ has one-dimensional range in a Hilbert-space representation. For two such projections $P,Q$ and any tracial state $\tau$ on $\mcA$, traciality gives
\begin{equation*}
  (\tr_N\otimes\tau)(Q)=\frac1N
\end{equation*}
and
\begin{equation}
  (\tr_N\otimes\tau)(PQ) =\frac1N\tau\bigl((\zeta_P^*\zeta_Q)^* (\zeta_P^*\zeta_Q)\bigr).
  \label{eq:pc-column-pair-trace}
\end{equation}
In particular, if $(\zeta_P^*\zeta_Q)^*(\zeta_P^*\zeta_Q)$ is a fixed scalar multiple of $\Id_{\mcA}$, then \Cref{eq:pc-column-pair-trace} implies that the value $(\tr_N\otimes\tau)(PQ)$ is independent of the choice of $\tau$.

\begin{lemma}
  \label{lem:pc-basic-projections}
  Let $\mcA$ be a unital $C^*$-algebra and $N\geq1$. For every $1\leq a\leq N$, let
  \begin{equation}
    q_a:=\ket{a}\bra{a}\otimes\Id_{\mcA}.
    \label{eq:pc-indexed-qx}
  \end{equation}
  For distinct $a,b\in\{1,\ldots,N\}$ and any unitary $h\in\mcA$, define
  \begin{equation}
    \begin{aligned}
      G_{a,b}(h)&:=\frac12\bigl(q_a+q_b +\ket{a}\bra{b}\otimes h+\ket{b}\bra{a}\otimes h^*\bigr),\\
      K_{a,b}(h)&:=\frac12(q_a+q_b) +\frac1{\sqrt2}\bigl(G_{a,b}(h)-q_b\bigr),\\
      D_{a,b}(h)&:=q_a+q_b-G_{a,b}(h).
    \end{aligned}
    \label{eq:pc-indexed-matrices}
  \end{equation}
  All these operators in $M_N(\mcA)$ are rank-one projections over $\mcA$. Moreover, for any column $\xi=\sum_{i=1}^{N}\ket{i}\otimes\xi_i\in M_{N,1}(\mcA)$, the equality $D_{a,b}(h)\xi=0$ if and only if $\xi_a=h\xi_b$.
\end{lemma}

\begin{proof}
  The normalized columns
  \begin{equation}
    \begin{aligned}
      \zeta_{q_a}&:=\ket{a}\otimes\Id_{\mcA},\\
      \zeta_{G_{a,b}(h)}&:=\frac1{\sqrt2} \bigl(\ket{a}\otimes h+\ket{b}\otimes\Id_{\mcA}\bigr),\\
      \zeta_{K_{a,b}(h)}&:=\cos(\pi/8)\ket{a}\otimes h +\sin(\pi/8)\ket{b}\otimes\Id_{\mcA},\\
      \zeta_{D_{a,b}(h)}&:=\frac1{\sqrt2} \bigl(\ket{a}\otimes h-\ket{b}\otimes\Id_{\mcA}\bigr)
    \end{aligned}
    \label{eq:pc-basic-columns}
  \end{equation}
  satisfy $Q=\zeta_Q\zeta_Q^*$ for their respective matrices $Q$.

  For the ``moreover" part, observe that $\zeta_{D_{a,b}(h)}^*\xi=\tfrac{1}{\sqrt{2}}(h^*\xi_a-\xi_b)$. Since $\zeta_{D_{a,b}(h)}$ is normalized, this vanishes if and only if $D_{a,b}(h)\xi=0$.
\end{proof}

\Cref{eq:pc-indexed-qx,eq:pc-indexed-matrices} in the above lemma construct projections from a given unitary, which will be used as generators for $M_N(\mcA)$ in \Cref{thm:projection-construction}. The next lemma gives the converse needed later for part (c) of \Cref{thm:projection-construction}. Using a linear identity between two projections, we can determine the partial isometry appearing in their matrix formulas.

\begin{lemma}
  \label{lem:pc-two-projections}
  Let $Q_A,Q_B,G,K$ be projections in a unital $C^*$-algebra satisfying
  \begin{equation*}
    Q_AQ_B=0,\quad (Q_A+Q_B)G=G,\quad (Q_A+Q_B)K=K,
  \end{equation*}
  and
  \begin{equation}
    K=\frac12(Q_A+Q_B)+\frac1{\sqrt2}(G-Q_B).
    \label{eq:pc-two-projection-relation}
  \end{equation}
  Then there is a unique partial isometry $x$ satisfying
  \begin{equation*}
    x^*x=Q_B,\quad xx^*=Q_A, ~\text{ and }~ G=\frac12(Q_A+Q_B+x+x^*).
  \end{equation*}
  It is given by $x=2Q_AGQ_B$.
\end{lemma}

\begin{proof}
  Taking adjoints also gives $G(Q_A+Q_B)=G$ and $K(Q_A+Q_B)=K$. Squaring \eqref{eq:pc-two-projection-relation} and using $K^2=K$ gives $(G-Q_B)^2=\tfrac{1}{2}(Q_A+Q_B)$. Multiplying on the left and right by $Q_A$ in one calculation and by $Q_B$ in another, and using $G^2=G$, gives
  \begin{equation*}
    Q_AGQ_A=\frac{1}{2}Q_A,~\text{ amnd }~ Q_BGQ_B=\frac{1}{2}Q_B.
  \end{equation*}
  For $x:=2Q_AGQ_B$, the identity $(Q_A+Q_B)G(Q_A+Q_B)=G$ now gives
  \begin{equation*}
    G=(Q_A+Q_B+x+x^*)/2.
  \end{equation*}
  Since $x=Q_AxQ_B$, multiplying $G^2=G$ on the left and right by each of $Q_A,Q_B$ gives
  \begin{equation*}
    \frac{1}{4}(Q_A+xx^*)=\frac{1}{2}Q_A ~\text{ and }~ \frac{1}{4}(Q_B+x^*x)=\frac{1}{2}Q_B.
  \end{equation*}
  Thus $xx^*=Q_A$ and $x^*x=Q_B$. The formula $x=2Q_AGQ_B$ also follows for any partial isometry with the stated properties, proving uniqueness.
\end{proof}

The next lemma uses equations between the entries of a column to determine its associated projection. These equations first show that any other projection satisfying them is a subprojection of the one given by the column. Equality of their traces then gives equality of the projections. The lemma also computes the pair values by comparing with matrices whose entries are scalar multiples of the identity.

\begin{lemma}
  \label{lem:pc-column-projection}
  Let $\mcA$ be a unital $C^*$-algebra, let $1\leq k\leq N$, and choose distinct indices $a_1,\ldots,a_k\in\{1,\ldots,N\}$. Set $c_1:=\Id_{\mcA}$. For every $2\leq j\leq k$, choose $1\leq\nu(j)<j$ and a unitary $h_j\in\mcA$, and recursively define $c_j:=h_jc_{\nu(j)}$. The following hold.
  \begin{enumerate}[(a)]
    \item For every faithful tracial state $\sigma$ on $\mcA$, there is a unique projection $E\in M_N(\mcA)$ satisfying
          \begin{equation}
            E=QEQ, ~~ (\tr_N\otimes\sigma)(E)=\frac1N, ~\text{ and }~ (\tr_N\otimes\sigma) \bigl(ED_{a_j,a_{\nu(j)}}(h_j)\bigr)=0 \label{eq:pc-column-conditions}
          \end{equation}
          for all $2\leq j\leq k$, where $Q:=\sum_{j=1}^kq_{a_j}$. It is given by $E=\eta\eta^*$, where
          \begin{equation}
            \eta:=\frac1{\sqrt{k}}\sum_{j=1}^k\ket{a_j}\otimes c_j.
            \label{eq:pc-unitary-column}
          \end{equation}

    \item Pick two arbitrary projections $P,Q$ in the family
          \begin{equation*}
            \begin{aligned}
              \mcP:= &\{q_a:1\leq a\leq N\}\cup\{E\}\\
              &\quad\cup\{G_{a_j,a_{\nu(j)}}(h_j), K_{a_j,a_{\nu(j)}}(h_j), D_{a_j,a_{\nu(j)}}(h_j):2\leq j\leq k\}\subseteq M_N(\mcA),
            \end{aligned}
          \end{equation*}
          where $q_a,G_{a,b}(h),K_{a,b}(h),D_{a,b}(h)$ are defined as in \Cref{lem:pc-basic-projections}. The value $(\tr_N\otimes\tau)(PQ)$ is independent of the choice of tracial state $\tau$ on $\mcA$.
  \end{enumerate}
\end{lemma}

\begin{proof}
  For part (a), we first define $E:=\eta\eta^*$ using \Cref{eq:pc-unitary-column} and examine that it satisfies all the conditions in \eqref{eq:pc-column-conditions}. Every $c_j$ is unitary, so $\eta^*\eta=\Id_\mcA$ and $E=\eta\eta^*$ is a rank-one projection over $\mcA$. It follows that $(\tr_N\otimes\tau)(E)=\tfrac{1}{N}$ for every tracial state $\tau$ on $\mcA$. The recursion $c_j:=h_jc_{\nu(j)}$ and the ``moreover" part of \Cref{lem:pc-basic-projections} give $D_{a_j,a_{\nu(j)}}(h_j)\eta=0$, so $E$ satisfies \eqref{eq:pc-column-conditions}.

  For uniqueness in~(a), fix a faithful tracial state $\sigma$ and let $P$ be a projection satisfying \eqref{eq:pc-column-conditions}. Faithfulness then implies
  \begin{equation}
    D_{a_j,a_{\nu(j)}}(h_j)P=0 \label{eq:DP=0}
  \end{equation}
  for all $2\leq j\leq k$. Writing $P_{a,b}\in\mcA$ for the matrix entries of $P\in M_N(\mcA)$, the $(a_j,b)$-entry of \Cref{eq:DP=0} gives
  \begin{equation*}
    0=\frac12\bigl(P_{a_j,b}-h_jP_{a_{\nu(j)},b}\bigr)
  \end{equation*}
  for all $2\leq j\leq k$ and $1\leq b\leq N$. We now show by induction on $j$ that $P_{a_j,b}=c_jP_{a_1,b}$ for every $b$. The case $j=1$ follows from $c_1=\Id_{\mcA}$. For $j\geq2$, since $\nu(j)<j$, the induction hypothesis on $\nu(j)$ and the equality $c_j=h_jc_{\nu(j)}$ give
  \begin{equation*}
    P_{a_j,b} =h_jP_{a_{\nu(j)},b} =h_jc_{\nu(j)}P_{a_1,b} =c_jP_{a_1,b}.
  \end{equation*}
  This completes the induction. Together with $P=QPQ$, these identities give $EP=P$. Indeed, for $1\leq i\leq k$ and every $b$,
  \begin{equation*}
    (EP)_{a_i,b}=\frac1k\sum_{j=1}^k c_ic_j^*P_{a_j,b} =c_iP_{a_1,b}=P_{a_i,b},
  \end{equation*}
  and all other rows are zero. Thus $P\leq E$. Both projections have trace $\tfrac{1}{N}$, so faithfulness gives $P=E$. This proves the uniqueness of $E$.

  For part~(b), for each $P\in\mcP$, let $P^{(0)}\in M_N(\C)$ be obtained from its defining formula by replacing every $h_j,c_j$, and $\Id_{\mcA}$ with $1$, without changing the indices. For instance,
  \begin{equation*}
    q_a^{(0)}=\ket{a}\bra{a},\quad E^{(0)}=\frac1k\sum_{i,j=1}^k\ket{a_i}\bra{a_j}.
  \end{equation*}
  Observe that the unitary
  \begin{equation*}
    W:=\sum_{j=1}^k\ket{a_j}\bra{a_j}\otimes c_j +\sum_{\substack{1\leq a\leq N\\a\notin\{a_1,\ldots,a_k\}}}q_a \in M_N(\mcA)
  \end{equation*}
  fixes each $q_a$ by conjugation and satisfies
  \begin{equation*}
    W^*EW=\frac1k\sum_{i,j=1}^k\ket{a_i}\bra{a_j}\otimes\Id=E^{(0)}\otimes\Id.
  \end{equation*}
  Moreover, since $c_j^*h_jc_{\nu(j)}=\Id$, we have
  \begin{align*}
    W^*&G_{a_j,a_{\nu(j)}}(h_j)W\\ & = \ket{a_j}\bra{a_j}\otimes\Id + \ket{a_{\nu(j)}}\bra{a_{\nu(j)}}\otimes\Id+ \ket{a_j}\bra{a_{\nu(j)}}\otimes c_j^*h_jc_{\nu(j)}+\ket{a_{\nu(j)}}\bra{a_j}\otimes c_{\nu(j)}^*h_j^*c_j \\
    &=G^{(0)}_{a_j,a_{\nu(j)}}(h_j)\otimes\Id.
  \end{align*}
  Similarly, $W^*K_{a_j,a_{\nu(j)}}(h_j)W =K^{(0)}_{a_j,a_{\nu(j)}}(h_j)\otimes\Id$ and $W^*D_{a_j,a_{\nu(j)}}(h_j)W =D^{(0)}_{a_j,a_{\nu(j)}}(h_j)\otimes\Id$. We conclude that $W^*PW=P^{(0)}\otimes\Id$ for every $P\in\mcP$. Traciality then gives
  \begin{equation*}
    (\tr_N\otimes\tau)(PQ) =(\tr_N\otimes\tau)(W^*PQW) =\tr_N(P^{(0)}Q^{(0)}),
  \end{equation*}
  for all $P,Q\in\mcP$. Since $\tr_N(P^{(0)}Q^{(0)})$ does not involve $\tau$, this value is independent of choice of $\tau$, which proves (b).
\end{proof}

We also need to determine projections whose columns have prescribed real scalar entries. The multiplier in an equation $z_a=tz_1$ need not be unitary, so we use the scalar projections below in place of $D_{a,1}(h)$. The resulting uniqueness statement will apply both to the projections expressing the relations and to certain scalar projections used to compare repeated generators.

\begin{lemma}
  \label{lem:pc-scalar-column}
  Let $\mcA$ be a unital $C^*$-algebra with a faithful tracial state $\sigma$, and let $N\geq1$. For $2\leq a\leq N$ and $t\in\R$, define
  \begin{equation}
    B_a(t):= \frac{q_a+t^2q_1 -t(2G_{a,1}(\Id)-q_a-q_1)}{1+t^2}.
    \label{eq:pc-weighted-comparison}
  \end{equation}
  Then $B_a(t)$ is a rank-one projection over $\mcA$.

  Choose $0\leq s\leq N-1$, distinct indices $e_1,\ldots,e_s\in\{2,\ldots, N\}$, and real numbers $\lambda_1,\ldots,\lambda_s$. Put
  \begin{equation*}
    Q:=q_1+\sum_{j=1}^s q_{e_j} ~\text{ and }~ \Lambda:=1+\sum_{j=1}^s\lambda_j^2.
  \end{equation*}
  There is a unique projection $F\in M_N(\mcA)$ satisfying
  \begin{equation}
    F=QFQ,~~(\tr_N\otimes\sigma)(F)=\frac1N, ~\text{ and }~ (\tr_N\otimes\sigma) \bigl(FB_{e_j}(\lambda_j)\bigr)=0 \label{eq:pc-scalar-column-conditions}
  \end{equation}
  for all $1\leq j\leq s$. It is given by $F=\xi\xi^*$, where
  \begin{equation}
    \xi:=\frac1{\sqrt{\Lambda}} \left(\ket{1}+\sum_{j=1}^s\lambda_j\ket{e_j}\right) \otimes\Id_{\mcA}.
    \label{eq:pc-weighted-scalar-column}
  \end{equation}
  In particular, $F$ is rank one over $\mcA$.
\end{lemma}

\begin{proof}
  The normalized column
  \begin{equation*}
    \zeta_{B_a(t)} :=\frac{(\ket{a}-t\ket{1})\otimes\Id_{\mcA}} {\sqrt{1+t^2}}
  \end{equation*}
  satisfies $B_a(t)=\zeta_{B_a(t)}\zeta_{B_a(t)}^*$, so $B_a(t)$ is rank one over $\mcA$.

  We first examine that the projection $F:=\xi\xi^*$ defined by \Cref{eq:pc-weighted-scalar-column} satisfies all the conditions in \eqref{eq:pc-scalar-column-conditions}. The first two conditions follow immediately from the definition of $F$. For the last one, note that for any column $z=\sum_{a=1}^N\ket{a}\otimes z_a\in M_{N,1}(\mcA)$, $B_a(t)z=0$ if and only if $z_a=tz_1$. Since $\xi$ is normalized and satisfies $B_{e_j}(\lambda_j)\xi=0$ for every $j$, we have $FB_{e_j}(\lambda_i)=0$. We conclude that $F$ satisfies all the conditions in \eqref{eq:pc-scalar-column-conditions}.

  For uniqueness, let $P$ be another projection satisfying these conditions. For each $j$, traciality and faithfulness give
  \begin{equation*}
    (\tr_N\otimes\sigma) \bigl((B_{e_j}(\lambda_j)P)^* (B_{e_j}(\lambda_j)P)\bigr) =(\tr_N\otimes\sigma) \bigl(PB_{e_j}(\lambda_j)\bigr)=0,
  \end{equation*}
  and hence $B_{e_j}(\lambda_j)P=0$. Taking its $(e_j,b)$-entry yields
  \begin{equation*}
    P_{e_j,b}=\lambda_jP_{1,b}
  \end{equation*}
  for all $1\leq j\leq s$ and $1\leq b\leq N$. Since $P=QPQ$, all other rows are zero. Consequently, the $b$-th column of $P$ equals $\sqrt{\Lambda}\,\xi P_{1,b}$. It follows that $FP=P$, so $P\leq F$. Both projections have trace $\tfrac{1}{N}$, so faithfulness gives $P=F$.
\end{proof}

The formulas for $B_a(t)$ and $F$ define projections even when $\mcA$ has no faithful trace; faithfulness is needed only for uniqueness. The projections $B_a(t)$ are used only through the real linear expression \eqref{eq:pc-weighted-comparison}; they are not additional members of the final generating family. Thus their pair values with a chosen projection are determined by the pair values involving $q_a,q_1$, and $G_{a,1}(\Id)$. The same applies to the projections $D_{a,b}(h)=q_a+q_b-G_{a,b}(h)$.

For distinct $a,b\in\{2,\ldots,N\}$, define
\begin{equation}
  \zeta_{C_{a,b}}:= \frac1{\sqrt3}(\ket{1}+\ket{a}+\ket{b})\otimes\Id_{\mcA} ~\text{ and }~ C_{a,b}:=\zeta_{C_{a,b}}\zeta_{C_{a,b}}^*.
  \label{eq:pc-scalar-column}
\end{equation}
This is the case of \Cref{lem:pc-scalar-column} with indices $a,b$ and weights $1,1$. In particular, its trace, its specified rows and columns, and its zero pair values with $B_a(1)$ and $B_b(1)$ determine it uniquely under a faithful trace. Its defining formula also gives
\begin{equation}
  \ket{a}\bra{b}\otimes\Id+\ket{b}\bra{a}\otimes\Id =3C_{a,b}+q_1-2G_{a,1}(\Id)-2G_{b,1}(\Id).
  \label{eq:pc-scalar-linear-expression}
\end{equation}
The right-hand side is a real linear combination of projections. This will allow us to use the two specified entries on the left in the proof that the canonical projections have the required images.

The above two lemmas allow us to encode a relation in $\mcA$ involving unitaries using projections in $M_N(\mcA)$. We illustrate this with the following small example.

\begin{example}
  \label{ex:pc-commutation}
  Let $u,v$ be unitaries in a unital $C^*$-algebra $\mcA$, and suppose we aim to encode the commutation relation $uv-vu=0$. In \Cref{lem:pc-column-projection}, take $N=6$, $k=5$, and $a_i=i+1$ for $1\leq i\leq5$. Set $c_1:=\Id$ and define the remaining entries by
  \begin{equation*}
    \begin{aligned}
      c_2:=vc_1=v,\quad c_3:=uc_2=uv, \quad c_4:=uc_1=u,\quad c_5:=vc_4=vu.
    \end{aligned}
  \end{equation*}
  The corresponding normalized column is
  \begin{equation*}
    \eta:=\frac1{\sqrt5} \begin{pmatrix}0\\\Id\\v\\uv\\u\\vu\end{pmatrix} \in M_{6,1}(\mcA).
  \end{equation*}
  To construct the scalar column $\xi$ in \Cref{lem:pc-scalar-column}, choose the indices $e_1=4$, $e_2=6$ and weights $\lambda_1=1$, $\lambda_2=-1$. This gives
  \begin{equation*}
    \xi:=\frac1{\sqrt3} \begin{pmatrix}\Id\\0\\0\\\Id\\0\\-\Id\end{pmatrix}.
  \end{equation*}
  Thus $\xi$ selects the entries $uv$ and $vu$ and direct computation yields
  \begin{equation}
    \xi^*\eta=\frac{uv-vu}{\sqrt{15}}.
    \label{eq:pc-commutation-inner-product}
  \end{equation}
  Hence $uv-vv=0$ if and only if $\eta$ and $\xi$ are orthogonal. As in \Cref{lem:pc-column-projection,lem:pc-scalar-column}, we define the projections $E:=\eta\eta^*$ and $F:=\xi\xi^*$ in $M_6(\mcA)$. Then $uv-vv=0$ if and only if $EF=0$.
\end{example}

For an algebra with multiple relations, we will use one column $\eta$ to store all $*$-monomials (and their suffixes) appearing in those relations. Each relation $R_j=0$ will have its own scalar column $\xi_j$, whose coefficients select the corresponding terms of $\eta$. Then $\xi_j^*\eta$ is a scalar multiple of $R_j$, so $E:=\eta\eta^*$ will satisfy $EF_j=0$ for $F_j:=\xi_j\xi_j^*$.

Different relations may use the same generator. In the proof of part~(c) in \Cref{thm:projection-construction}, we must show that all these uses give the same unitary in the GNS representation of any tracial state. The next lemma provides the required equality.

\begin{lemma}
  \label{lem:pc-equal-coefficients}
  For unitaries $h,k$ in a unital $C^*$-algebra $\mcA$, define
  \begin{equation}
    T(h,k):=\frac12\Id_4\otimes\Id_{\mcA} +\frac1{2\sqrt2}
    \begin{pmatrix}
      0 & h & \Id & 0 \\
      h^* & 0 & 0 & \Id \\
      \Id & 0 & 0 & -k \\
      0 & \Id & -k^* & 0
    \end{pmatrix}.
    \label{eq:pc-four-row-projection}
  \end{equation}
  Then $T(h,k)$ is a projection if and only if $h=k$.
\end{lemma}

\begin{proof}
  The matrix is self-adjoint, and direct multiplication gives
  \begin{equation*}
    T(h,k)^2-T(h,k)=\frac18
    \begin{pmatrix}
      0 & 0 & 0 & h-k \\
      0 & 0 & h^*-k^* & 0 \\
      0 & h-k & 0 & 0 \\
      h^*-k^* & 0 & 0 & 0
    \end{pmatrix}.
  \end{equation*}
  This vanishes exactly when $h=k$.
\end{proof}

\subsection{Construction of the generating family}
\label{subsec:pc-construction}

Suppose $\mcA$ has a finite real-coefficient unitary presentation and admits a tracial state. Throughout the rest of this section, fix a presentation
\begin{equation*}
  \mcA=C^*\ang{u_1,\ldots,u_d\text{ unitary} \mid R_1=0,\ldots,R_m=0},
\end{equation*}
where $R_j=\sum_{w\in\mcS}\lambda_{j,w}w$ as in \Cref{def:real-coefficient-unitary}.

For each $*$-monomial $w=h_1\cdots h_L\in\mcS$, where each $h_i\in\{u_1,\ldots,u_d,u_1^*,\ldots,u_d^*\}$, consider the finite set
\begin{equation*}
  \{\Id,h_L,h_{L-1}h_L,\ldots,h_1\cdots h_L\}.
\end{equation*}
Let $\mcW$ be the union of these sets. For example, the monomial $u_1u_2^*u_3$ contributes $\Id,u_3,u_2^*u_3,u_1u_2^*u_3$. Thus $\mcW$ is finite and contains $\mcS$. Note that every $*$-monomial $w\neq \Id$ in $\mcW$ has a unique formal expression $hv$, where $h$ is its first variable and $v\in\mcW$ is the product of the remaining variables ($v=\Id$ if $\abs{w}=1$). Put
\begin{equation}
  k:=|\mcW|+|\mcS| ~\text{ and }~ N:=2d+3+k.
  \label{eq:pc-matrix-size}
\end{equation}
For convenience, label the indices $\{1,\ldots,N\}$ as follows. Reserve index $1$. Write
\begin{equation}
  s_i:=2i+2 ~\text{ and }~ t_i:=2i+3
  \label{eq:pc-generator-indices}
\end{equation}
for all $0\leq i\leq d$, so they label the indices $\{2,\ldots,2d+3\}$. Label the remaining $k$ indices by
\begin{equation*}
  a_w,w\in\mcW ~\text{ and }~b_w,w\in\mcS.
\end{equation*}
All these indices are distinct. In particular, $a_w\ne b_w$ for $w\in\mcS$. Denote the set of these $k$ indices by $\mcV$.

We further select the triples
\begin{equation}
  \begin{aligned}
    \mcI:={}&\{(a_{hv},a_v,h):hv\in\mcW, \ h\in\{u_1,\ldots,u_d,u_1^*,\ldots,u_d^*\}\}\\
    &\cup\{(b_w,a_w,\Id):w\in\mcS\}.
  \end{aligned}
  \label{eq:pc-entry-triples}
\end{equation}
Such a triple $(a,b,h)$ will be used to record the equation $z_a=hz_b$ for a column $z$. Note that no two distinct triples have the same (unordered) pair $\{a,b\}$, and at most one of $a,b$ is of the form $b_w$.

Throughout the rest of this section, we use projections $q_a,G_{a,b}(h),K_{a,b}(h),D_{a,b}(h),C_{a,b}\in M_N(\mcA)$ as defined in \Cref{lem:pc-basic-projections} and \Cref{eq:pc-scalar-column}.

\begin{lemma}
  \label{lem:pc-relation-projections}
  With the above indices and notations, define
  \begin{equation}
    \eta:=\frac1{\sqrt{k}}\left( \sum_{w\in\mcW}\ket{a_w}\otimes w +\sum_{w\in\mcS}\ket{b_w}\otimes w\right)\in M_{N,1}(\mcA) ~\text{ and }~ E:=\eta\eta^*\in M_N(\mcA).
    \label{eq:pc-shared-columns}
  \end{equation}
  For each $1\leq j\leq m$, put $\Lambda_j:=1+\sum_{w\in\mcS}\lambda_{j,w}^2$ and define
  \begin{equation*}
    \xi_j:=\frac1{\sqrt{\Lambda_j}}\left( \ket{1}+\sum_{w\in\mcS}\lambda_{j,w}\ket{b_w}\right) \otimes\Id\in M_{N,1}(\mcA) ~\text{ and }~ F_j:=\xi_j\xi_j^*\in M_N(\mcA).
  \end{equation*}
  Then $E,F_1,\ldots,F_m$ are rank-one projections over $\mcA$, and
  \begin{equation}
    \xi_j^*\eta=\frac{R_j}{\sqrt{k\Lambda_j}}=0 ~\text{ and }~ EF_j=0
    \label{eq:pc-shared-overlap}
  \end{equation}
  for all $1\leq j\leq m$. Moreover, for any $P,Q$ in
  \begin{equation*}
    \begin{aligned}
      &\{q_a:1\leq a\leq N\}\cup\{E,F_1,\ldots,F_m\}\\
      &\quad\cup\{G_{a,b}(h),K_{a,b}(h),D_{a,b}(h): (a,b,h)\in\mcI\},
    \end{aligned}
  \end{equation*}
  the value $(\tr_N\otimes\tau)(PQ)$ is independent of the choice of the tracial state $\tau$ on $\mcA$.
\end{lemma}

\begin{proof}
  Unitarity of the monomials and the definition of $\Lambda_j$ give $\eta^*\eta=\xi_j^*\xi_j=\Id$. Only the $b_w$-rows of $\xi_j^*$ contribute to $\xi_j^*\eta$, which precisely gives $\tfrac{1}{\sqrt{k\Lambda_j}}\sum_{w\in\mcS}\lambda_{j,w}w=\tfrac{1}{\sqrt{k\Lambda_j}}R_j$. This proves \eqref{eq:pc-shared-overlap}.

  It remains to prove the ``moreover'' part. We wish to apply \Cref{lem:pc-column-projection}. For this, we enumerate $\mcV$ as $r_1,\ldots,r_k$ by first listing the indices $a_w$ in nondecreasing order of the length of $w$, and then all the indices $b_w$. Monomials of equal length and the indices $b_w$ may be ordered arbitrarily. In particular, $r_1=a_{\Id}$. We can then rewrite
  \begin{equation*}
    \eta=\frac1{\sqrt{k}}\sum_{i=1}^k\ket{r_i}\otimes c_i.
  \end{equation*}
  For each $2\leq i\leq k$, let $(r_i,r_{\nu(i)},h_i)$ be the unique triple in $\mcI$ whose first index is $r_i$. If $r_i=a_{hv}$, its second index is $a_v$, which occurs earlier because $v$ has one fewer factor than $hv$. If $r_i=b_w$, its second index is $a_w$, which also occurs earlier. Consequently $\nu(i)<i$,
  \begin{equation}
    c_i=h_ic_{\nu(i)}
    \label{eq:pc-shared-entry-equations}
  \end{equation}
  for all $2\leq i\leq k$, and $c_1=\Id$. Then \Cref{lem:pc-column-projection}(b) implies that $(\tr_N\otimes \tau)(PQ)$ is independent of $\tau$ for all pairs $P,Q$ not involving an $F_j$.

  For the remaining pairs involving $F_j$, fix $1\leq j\leq m$ and an arbitrary tracial state $\tau$ on $\mcA$. There are two types of triples in $\mcI$. For a triple $(a_{hv},a_v,h)$, the entries of $\xi_j$ at both $a_{hv}$ and $a_v$ are zero. Therefore
  \begin{equation*}
    \begin{aligned}
      (\tr_N\otimes\tau)(F_jG_{a_{hv},a_v}(h)) &=(\tr_N\otimes\tau)(F_jK_{a_{hv},a_v}(h))\\
      &=(\tr_N\otimes\tau)(F_jD_{a_{hv},a_v}(h))=0.
    \end{aligned}
  \end{equation*}
  Every other triple has the form $(b_w,a_w,\Id)$ for some $w\in\mcS$. For these, we have
  \begin{equation*}
    \begin{aligned}
      (\tr_N\otimes\tau)(F_jG_{b_w,a_w}(\Id)) &=(\tr_N\otimes\tau)(F_jD_{b_w,a_w}(\Id)) =\frac{\lambda_{j,w}^2}{2N\Lambda_j},\\
      (\tr_N\otimes\tau)(F_jK_{b_w,a_w}(\Id)) &=\frac{\lambda_{j,w}^2\cos^2(\pi/8)} {N\Lambda_j}
    \end{aligned}
  \end{equation*}
  for all $w\in\mcS$. The entries of $\xi_j$ also give
  \begin{equation*}
    (\tr_N\otimes\tau)(F_jq_r) =
    \begin{cases}
      \dfrac1{N\Lambda_j}, & r=1, \\[4pt]
      \dfrac{\lambda_{j,w}^2}{N\Lambda_j}, & r=b_w\text{ for some }w\in\mcS, \\[4pt]
      0, & \text{otherwise}
    \end{cases}
  \end{equation*}
  Finally,
  \begin{equation}
    (\tr_N\otimes\tau)(F_jF_{j'}) =\frac{\left(1+\sum_{w\in\mcS} \lambda_{j,w}\lambda_{j',w}\right)^2} {N\Lambda_j\Lambda_{j'}}
    \label{eq:pc-scalar-relation-pairs}
  \end{equation}
  for all $1\leq j'\leq m$. Together with $EF_j=0$, we conclude that (b) holds.
\end{proof}

To construct the desired generating set for \Cref{thm:projection-construction}, we first define the common family
\begin{equation}
  \begin{aligned}
    \mcP_0:={}&\{q_a:1\leq a\leq N\} \cup\bigcup_{a=2}^N\{G_{a,1}(\Id),K_{a,1}(\Id)\}\\
    &\cup\bigcup_{i=0}^d \{G_{s_i,t_i}(u_i),K_{s_i,t_i}(u_i)\} \cup\{C_{s_0,t_0}\}.
  \end{aligned}
  \label{eq:pc-common-family}
\end{equation}
For notational convenience, we let
\begin{equation*}
  s(\Id):=s_0,~~ t(\Id):=t_0,~~ (s(u_i),t(u_i)):=(s_i,t_i), ~\text{ and }~ (s(u_i^*),t(u_i^*)):=(t_i,s_i).
\end{equation*}
for $1\leq i\leq d$. For each $(a,b,h)\in\mcI$, define $T_{a,b}(h)\in M_N(\mcA)$ by letting its submatrix on the ordered rows and columns $s(h),t(h),a,b$ to be $T(h,h)$ given by \Cref{lem:pc-equal-coefficients}, with all other entries zero. These four indices are distinct, so \Cref{lem:pc-equal-coefficients} makes $T_{a,b}(h)$ a projection. Set
\begin{equation}
  \begin{aligned}
    \mcP:={}&\mcP_0\cup\{E,F_1,\ldots,F_m\}\\
    &\cup\bigcup_{(a,b,h)\in\mcI} \{G_{a,b}(h),K_{a,b}(h), C_{s(h),a},C_{t(h),b},T_{a,b}(h)\}.
  \end{aligned}
  \label{eq:pc-final-family}
\end{equation}
We regard this as an indexed family: different names may denote equal matrices, while each requested $C_{a,b}$ is included only once.

For later use, note that the formulas for $G$ and \eqref{eq:pc-scalar-linear-expression} give
\begin{equation}
  \begin{aligned}
    T_{a,b}(h)={}&\frac12(q_{s(h)}+q_{t(h)}+q_a+q_b)\\
    &+\frac1{2\sqrt2}\Bigl( 2G_{s(h),{t(h)}}(h)-2G_{a,b}(h)-q_{s(h)}-q_{t(h)}+q_a+q_b\\
    &\quad+3C_{s(h),a}+3C_{{t(h)},b}+2q_1 -2\sum_{r\in\{s(h),{t(h)},a,b\}}G_{r,1}(\Id)\Bigr).
  \end{aligned}
  \label{eq:pc-comparison-linear}
\end{equation}
In other words, every $T_{a,b}(h)$ is a real linear combination of the other projections in $\mcP$.

\subsection{Proof of Theorem \ref{thm:projection-construction}}
Let $X$ index the family \eqref{eq:pc-final-family}, and write its members as $P_x$. We are now ready to prove that $\{P_x:x\in X\}$ is the desired generating set for \Cref{thm:projection-construction}.
\begin{proof}[Proof of \Cref{thm:projection-construction}]
  Universality gives a unital $*$-homomorphism
  \begin{equation*}
    \Phi:\PVM{X,\Z_2}\arr M_N(\mcA) ~\text{ sending }~ e_{x,1}\mapsto P_x.
  \end{equation*}

  For part (a), we show that the subset $\mcP_0$ given in \eqref{eq:pc-common-family} already generates $M_N(\mcA)$. In fact,
  \begin{equation*}
    \begin{aligned}
      2q_aG_{a,1}(\Id)q_1&=\ket{a}\bra{1}\otimes\Id,\\
      2q_{s_i}G_{s_i,t_i}(u_i)q_{t_i} &=\ket{s_i}\bra{t_i}\otimes u_i
    \end{aligned}
  \end{equation*}
  for all $2\leq a\leq N$ and $1\leq i\leq d$. Products and adjoints of $\{\ket{a}\bra{1}\otimes\Id:1\leq a\leq N\}$ give every $\ket{a}\bra{b}\otimes\Id$. Multiplying $\{\ket{s_i}\bra{t_i}\otimes u_i:0\leq i\leq d \}$ by these matrices gives all $\ket{1}\bra{1}\otimes u_i$. Hence $\mcP_0$, and therefore $\mcP$, generates $M_N(\mcA)$.

  For part~(b), note that all $K_{a,b}(h)$ and $T_{a,b}(h)$ are real linear combinations of the remaining projections in $\mcP$, as shown in \eqref{eq:pc-indexed-matrices} and \eqref{eq:pc-comparison-linear}, so it suffices to check products of $q_a,G_{a,b}(h),C_{a,b},E$, and $F_j$. By \Cref{lem:pc-relation-projections}, $(\tr_N\otimes\tau)(PQ)$ for all $P,Q$ in
  \begin{equation*}
    \{q_a:1\leq a\leq N\}\cup\{E,F_1,\ldots,F_m\} \cup\{G_{a,b}(h):(a,b,h)\in\mcI\}
  \end{equation*}
  are already independent of $\tau$. We calculate only the remaining products. Fix an arbitrary tracial state $\tau$ on $\mcA$.
  \begin{itemize}
    \item \emph{Products with $q_r$ and products of two $G$-projections.} For the products with $q_r$ not covered by \Cref{lem:pc-relation-projections},
          \begin{equation*}
            \begin{aligned}
              (\tr_N\otimes\tau)(q_rG_{a,b}(h)) &=\frac{\delta_{ra}+\delta_{rb}}{2N},\\
              (\tr_N\otimes\tau)(q_rC_{a,b}) &=\frac{\delta_{r1}+\delta_{ra}+\delta_{rb}}{3N}
            \end{aligned}
          \end{equation*}
          for all $1\leq r\leq N$. For two $G$-projections, at least one of which is not indexed by $\mcI$,
          \begin{equation*}
            (\tr_N\otimes\tau)(G_{a,b}(h)G_{r,s}(g)) =\frac{|\{a,b\}\cap\{r,s\}|^2}{4N}.
          \end{equation*}

    \item \emph{The remaining products with $E$.} For $2\leq r\leq N$,
          \begin{equation*}
            (\tr_N\otimes\tau)(EG_{r,1}(\Id)) =\begin{cases}
              \frac{1}{2Nk}, & r\in\mcV, \\
              0, & r\notin\mcV.
            \end{cases}
          \end{equation*}
          For $0\leq i\leq d$,
          \begin{equation*}
            (\tr_N\otimes\tau)(EG_{s_i,t_i}(u_i))=0.
          \end{equation*}
          Finally, for every $C_{a,b}\in\mcP$,
          \begin{equation*}
            (\tr_N\otimes\tau)(EC_{a,b}) =\frac{|\{a,b\}\cap\mcV|}{3Nk}.
          \end{equation*}

    \item \emph{The remaining products with $F_j$.} Fix $1\leq j\leq m$. For $2\leq r\leq N$,
          \begin{equation*}
            (\tr_N\otimes\tau)(F_jG_{r,1}(\Id)) =\begin{cases}
              \dfrac{(1+\lambda_{j,w})^2}{2N\Lambda_j}, & r=b_w\text{ for some }w\in\mcS, \\[4pt]
              \dfrac1{2N\Lambda_j}, & \text{otherwise}.
            \end{cases}
          \end{equation*}
          For $0\leq i\leq d$,
          \begin{equation*}
            (\tr_N\otimes\tau)(F_jG_{s_i,t_i}(u_i))=0.
          \end{equation*}
          For every $C_{a,b}\in\mcP$,
          \begin{equation*}
            (\tr_N\otimes\tau)(F_jC_{a,b}) =\frac{\left(1+\sum_{\substack{w\in\mcS\\
                  b_w\in\{a,b\}}}\lambda_{j,w}\right)^2} {3N\Lambda_j}.
          \end{equation*}

    \item \emph{The remaining products with $C_{a,b}$.} For every $G_{r,s}(h)\in\mcP$,
          \begin{equation*}
            (\tr_N\otimes\tau)(C_{a,b}G_{r,s}(h)) =\frac{|\{1,a,b\}\cap\{r,s\}|^2}{6N}.
          \end{equation*}
          For any two $C$-projections,
          \begin{equation*}
            (\tr_N\otimes\tau)(C_{a,b}C_{r,s}) =\frac{\bigl(1+|\{a,b\}\cap\{r,s\}|\bigr)^2}{9N}.
          \end{equation*}
  \end{itemize}
  Together with \Cref{lem:pc-relation-projections}, this covers all the cases. We conclude that, for every $x,y\in X$, the value $(\tr_N\otimes\tau)(P_xP_y)$ is independent of $\tau$. Denote this value by $p_{x,y}$. It is real since
  \begin{equation*}
    (\tr_N\otimes\tau)(P_xP_y) =(\tr_N\otimes\tau)(P_xP_yP_x)\geq0.
  \end{equation*}
  This proves~(b).

  For part~(c), let $\tau$ be a tracial state on $\PVM{X,\Z_2}$ satisfying $\tau(e_{x,1}e_{y,1})=p_{x,y}$ for all $x,y\in X$. Let $(\pi_\tau,\mcH_\tau,\ket{\tau})$ be its GNS representation, and let
  \begin{equation*}
    \mcM:=\pi_\tau(\PVM{X,\Z_2})''\subseteq \msB(\mcH_\tau).
  \end{equation*}
  Then $\tau$ induces a faithful normal tracial state $\tau_{\mcM}$ on $\mcM$ such that $\tau=\tau_{\mcM}\circ\pi_\tau$. Put $P'_x:=\pi_\tau(e_{x,1})$ for all $x\in X$. For the projections $q_a,G_{a,b}(h),\ldots$ in $\mcP$, write $q'_a,G'_{a,b}(h),\ldots$ for the corresponding projections $P'_x$.

  Let $\mcC:=q'_1\mcM q'_1$. We will construct a $*$-isomorphism $\Psi:M_N(\mcC)\arr\mcM$ and a unital $*$-homomorphism $\rho:\mcA\arr\mcC$ such that
  \begin{equation*}
    \pi_\tau=\Psi\circ(\operatorname{id}_{M_N}\otimes\rho)\circ\Phi.
  \end{equation*}
  Taking traces will then give the desired factorization in~(c).

  For this, we first fix a tracial state $\varphi$ on $\mcA$. By part~(b),
  \begin{equation}
    \tau_{\mcM}(P'_xP'_y) =\tau(e_{x,1}e_{y,1}) =p_{x,y} =(\tr_N\otimes\varphi)(P_xP_y)
    \label{eq:pc-gns-product-traces}
  \end{equation}
  for all $x,y\in X$. Observable that, for any real numbers $\gamma_0,\gamma_x,x\in X$, and $L:=\gamma_0\Id+\sum_x\gamma_xP_x$ and $L':=\gamma_0\Id+\sum_x\gamma_xP'_x$, \Cref{eq:pc-gns-product-traces} further implies
  \begin{equation*}
    \tau_{\mcM}((L')^*L')=(\tr_N\otimes\varphi)(L^*L).
  \end{equation*}
  Faithfulness of $\tau_\mcM$ therefore concludes $L'=0$ whenever $L=0$. Applying these observations to $\sum_aq_a=\Id$ and $q_aq_b=0$ for $a\ne b$ shows that the $\{q'_a:1\leq a\leq N\}\subseteq \mcM$ are orthogonal and sum to $\Id_{\mcH_\tau}$. For any subset of indices $J\subseteq\{1,\ldots,N\}$, if $P\in\mcP$ is zero outside rows and columns indexed $J$, then
  \begin{equation*}
    \tau_{\mcM}(P'q'_a)=(\tr_N\otimes\varphi)(Pq_a)=0
  \end{equation*}
  for all $a\notin J$. Hence
  \begin{equation*}
    P'=\left(\sum_{a\in J}q'_a\right)P' \left(\sum_{a\in J}q'_a\right).
  \end{equation*}
  Set $v_1:=q'_1$ and $v_a:=2q'_aG'_{a,1}(\Id)q'_1$ for $2\leq a\leq N$. Applying \Cref{lem:pc-two-projections} to $G'_{a,1}(\Id)$ and $K'_{a,1}(\Id)$ gives $v_a^*v_a=q'_1$ and $v_av_a^*=q'_a$. Consequently,
  \begin{equation*}
    v_a^*v_b=\delta_{a,b}q'_1 ~\text{ and }~ \sum_{a=1}^Nv_av_a^*=\Id
  \end{equation*}
  for all $1\leq a,b\leq N$. Now define
  \begin{equation}
    \Psi:M_N(\mcC)\arr\mcM ~\text{ sending }~ \ket{a}\bra{b}\otimes c\mapsto v_acv_b^*.
    \label{eq:pc-matrix-isomorphism}
  \end{equation}
  This is a unital $*$-isomorphism, with inverse $y\mapsto(v_a^*yv_b)_{a,b=1}^N$. Since $\tau_{\mcM}(q'_1)=\tfrac{1}{N}$, the functional $\sigma:=N\cdot\tau_{\mcM}|_{\mcC}$ is a faithful normal tracial state on $\mcC$, and
  \begin{equation}
    \tau_{\mcM}\circ\Psi=\tr_N\otimes\sigma.
    \label{eq:pc-matrix-trace}
  \end{equation}
  Indeed, $\tau_{\mcM}(v_acv_b^*)=\tfrac{\delta_{a,b}\sigma(c)}{N}$.

  For $P\in\mcP$, write $\widehat P:=\Psi^{-1}(P')$. Then
  \begin{equation}
    (\tr_N\otimes\sigma)(\widehat P\widehat Q) =(\tr_N\otimes\varphi)(PQ)
    \label{eq:pc-hat-product-traces}
  \end{equation}
  for all $P,Q\in\mcP$, also for their real linear combinations with $\Id$. The definition of $\Psi$ gives
  \begin{equation*}
    \widehat q_a=\ket{a}\bra{a}\otimes\Id_{\mcC} ~\text{ and }~ \widehat G_{a,1}(\Id)=G_{a,1}(\Id_{\mcC})
  \end{equation*}
  for $1\leq a\leq N$ and $2\leq a\leq N$, respectively. Form $B_a(t)$ in $M_N(\mcC)$ using these scalar matrices. Its linear formula and \eqref{eq:pc-hat-product-traces} give
  \begin{equation*}
    \widehat C_{a,b}B_a(1)=\widehat C_{a,b}B_b(1)=0 ~\text{ and }~ \widehat F_jB_{b_w}(\lambda_{j,w})=0
  \end{equation*}
  by faithfulness, for all $C_{a,b}\in\mcP$, $1\leq j\leq m$, and $w\in\mcS$. The projections $\widehat C_{a,b},\widehat F_j$ have trace $1/N$ and are zero outside the same indices as $C_{a,b},F_j$. By \Cref{lem:pc-scalar-column}, each $\widehat C_{a,b}$ has the formula \eqref{eq:pc-scalar-column} over $\mcC$, and $\widehat F_j=\xi_{j,\mcC}\xi_{j,\mcC}^*$, where $\xi_{j,\mcC}$ is $\xi_j$ with $\Id_{\mcA}$ replaced by $\Id_{\mcC}$.

  The identity \eqref{eq:pc-scalar-linear-expression} also gives
  \begin{equation*}
    G_{s_0,t_0}(\Id) =\frac12(q_1+q_{s_0}+q_{t_0}+3C_{s_0,t_0}) -G_{s_0,1}(\Id)-G_{t_0,1}(\Id).
  \end{equation*}
  The same linear identity holds for the hatted projections, so $\widehat G_{s_0,t_0}(\Id)=G_{s_0,t_0}(\Id_{\mcC})$. For $1\leq i\leq d$, \Cref{lem:pc-two-projections} applied to $\widehat G_{s_i,t_i}(u_i),\widehat K_{s_i,t_i}(u_i)$ gives a unitary $U_i\in\mcC$ with $\widehat G_{s_i,t_i}(u_i)=G_{s_i,t_i}(U_i)$. Put $U:=(U_1,\ldots,U_d)$. Reversing the assigned pair for an adjoint shows that $\widehat G_{s(h),t(h)}(h)$ equals $G_{s(h),t(h)}(h(U))$ for every $h$ used in $\mcI$.

  For $(a,b,h)\in\mcI$, \Cref{lem:pc-two-projections} initially gives a unitary $H\in\mcC$ such that $\widehat G_{a,b}(h)=G_{a,b}(H)$. In \eqref{eq:pc-comparison-linear}, all scalar terms have been determined. Replacing its projections by their hatted images therefore makes the submatrix of $\widehat T_{a,b}(h)$ on the ordered indices $s(h),t(h),a,b$ equal to $T(h(U),H)$, with zero entries elsewhere. Since it is a projection, \Cref{lem:pc-equal-coefficients} gives $H=h(U)$. Thus $\widehat G_{a,b}(h)=G_{a,b}(h(U))$ for every $(a,b,h)\in\mcI$.

  Let $\eta_U$ be obtained from $\eta$ by replacing every monomial $w$ with $w(U)$. The projection $\widehat E$ has trace $1/N$ and is zero outside rows and columns in $\mcV$. Since $D_{a,b}(h)=q_a+q_b-G_{a,b}(h)$ and $ED_{a,b}(h)=0$, \eqref{eq:pc-hat-product-traces} and faithfulness give
  \begin{equation*}
    D_{a,b}(h(U))\widehat E=0
  \end{equation*}
  for every $(a,b,h)\in\mcI$. Using the enumeration of $\mcV$ from the proof of \Cref{lem:pc-relation-projections}, \Cref{lem:pc-column-projection}(a) applies to the same recursion evaluated at $U$, and gives
  \begin{equation}
    \widehat E=\eta_U\eta_U^*.
    \label{eq:pc-recovered-common-column}
  \end{equation}
  Likewise, $EF_j=0$ implies $\widehat F_j\widehat E=0$, and hence
  \begin{equation*}
    0=\xi_{j,\mcC}^*\eta_U =\frac{\sum_{w\in\mcS}\lambda_{j,w}w(U)}{\sqrt{k\Lambda_j}}
  \end{equation*}
  for all $1\leq j\leq m$. Universality therefore gives a unital $*$-homomorphism $\rho:\mcA\arr\mcC$ with $\rho(u_i)=U_i$ for all $1\leq i\leq d$.

  We have proved $\widehat P=(\operatorname{id}_{M_N}\otimes\rho)(P)$ for all the $q,G,C,E,F$ projections in $\mcP$. The same holds for $K,T$ by their linear formulas, so it holds for every $P\in\mcP$. On the generators $e_{x,1}$, this gives the required equality $\pi_\tau=\Psi\circ(\operatorname{id}_{M_N}\otimes\rho)\circ\Phi$. Putting $\tau_{\mcA}:=\sigma\circ\rho$ and taking traces yields
  \begin{equation*}
    \tau=(\tr_N\otimes\tau_{\mcA})\circ\Phi.
  \end{equation*}
  Surjectivity of $\Phi$ makes the amplified trace unique; evaluating on $\Id_N\otimes a$ for $a\in\mcA$ gives uniqueness of $\tau_{\mcA}$. This proves~(c).
\end{proof}

\section{Proof of Theorem~\ref{thm:main}}
\label{sec:proof-main}

\Cref{thm:projection-construction} reduces \Cref{thm:main} to finding an algebra with a unique finite-dimensional tracial state and a distinct amenable tracial state. Such an algebra exists. We first state a general result.

\begin{theorem}
  \label{thm:nonrobust-from-traces}
  Let $\mcA$ have a finite real-coefficient unitary presentation. Suppose $\mcA$ has a character $\tau_0$ which is its unique finite-dimensional tracial state, and an amenable tracial state $\tau_1\ne\tau_0$. Let $N,X,\{P_x:x\in X\}$, and $\Phi$ be given by \Cref{thm:projection-construction}. Define the tracial state $\wtd{\tau}_0:=(\tr_N\otimes\tau_0)\circ\Phi$ on $\PVM{X,\Z_2}$ and let
  \begin{equation}
    \wtd p(a,b\mid x,y):=\wtd{\tau}_0(e_{x,a}e_{y,b})
    \label{eq:general-nonrobust-correlation}
  \end{equation}
  for all $x,y\in X$ and $a,b\in\Z_2$. Then $\wtd p\in C_q^s(X,\Z_2)$ is extreme in $C_q(X,X,\Z_2,\Z_2)$. It is a self-test with an ideal strategy using a maximally entangled state of Schmidt rank $N$, but it is not a robust self-test.
\end{theorem}

\begin{proof}
  Let $\wtd{\tau}_1:=(\tr_N\otimes\tau_1)\circ\Phi$. By \Cref{lemma:pre_traces}, $\wtd{\tau}_0$ is finite-dimensional and $\wtd{\tau}_1$ is amenable. Thus \Cref{thm:synchronous-quantum-fd-trace} implies $\wtd p\in C_q^s(X,\Z_2)$. By \Cref{thm:projection-construction}(b), both pullbacks agree on $e_{x,1}e_{y,1}$. Taking $x=y$ gives their values on $e_{x,1}$, and $e_{x,0}=\Id-e_{x,1}$ determines all other outcome pairs. Consequently, both tracial states implement $\wtd p$.

  If a finite-dimensional tracial state $\varphi$ implements $\wtd p$, \Cref{thm:projection-construction}(c) gives
  \begin{equation*}
    \varphi=(\tr_N\otimes\tau)\circ\Phi
  \end{equation*}
  for a tracial state $\tau$ on $\mcA$. By \Cref{lemma:pre_traces}, $\tau$ is finite-dimensional, and hence $\tau=\tau_0$ by hypothesis. Therefore $\wtd{\tau}_0$ is the unique finite-dimensional tracial state implementing $\wtd p$.

  Since $\tau_0$ is a character, the map
  \begin{equation}
    \Phi_0:=(\operatorname{id}_{M_N}\otimes\tau_0)\circ\Phi :\PVM{X,\Z_2}\arr M_N
    \label{eq:general-ideal-representation}
  \end{equation}
  is a surjective $*$-homomorphism and $\wtd{\tau}_0=\tr_N\circ\Phi_0$. By \Cref{lem:matrix-trace-extreme}, $\wtd{\tau}_0$ is extreme among finite-dimensional tracial states. Its uniqueness as an implementing trace and \Cref{lem:unique-extreme-trace-implies-extreme-correlation} show that $\wtd p$ is extreme in the full quantum correlation set.

  The trace $\wtd{\tau}_0$ is also amenable, since it is finite-dimensional. The two amenable implementing traces are distinct: equality of their pullbacks, together with surjectivity of $\Phi$, would imply $\tau_0=\tau_1$. Thus \Cref{thm:tracial-self-testing} shows that $\wtd p$ is a self-test but not a robust self-test.

  Since $\tau_0$ is a character and $\Phi_0$ is a surjective $*$-homomorphism, the GNS representation of $\wtd{\tau}_0=\tr_N\circ\Phi_0$ is unitarily equivalent to
  \begin{equation*}
    \bigl(\C^N\otimes\C^N, \Phi_0(\,\cdot\,)\otimes\Id_N, \ket{\varphi_N} \bigr),
  \end{equation*}
  where $\ket{\varphi_N}=\frac{1}{\sqrt{N}}\sum_{i=1}^N\ket{i}\otimes\ket{i}$ is a maximally entangled state in $\C^N\otimes\C^N$. We have shown that $\wtd{\tau}_0$ is extreme and is the unique finite-dimensional tracial state implementing $\wtd p$. Hence \Cref{thm:ideal-strategy-from-trace} gives an ideal strategy $\wtd S$ for the self-test with shared state $\ket{\varphi_N}$, which has Schmidt rank $N$.
\end{proof}

\subsection{An explicit example from rotation algebras}
\label{sec:rotation-example}

Let $\mcA_{\mathrm{rot}}$ be the universal unital $C^*$-algebra generated by unitaries $z,u,v$ subject to
\begin{equation}
  zu=uz,\quad zv=vz,\quad uv=zvu,\quad (2z+2z^*-\Id)(u+v-2\Id)=0.
  \label{eq:rot-presentation}
\end{equation}
These are four real-coefficient $*$-polynomial relations. The assignment $z,u,v\mapsto1$ satisfies them and defines a character $\tau_0$; in particular, $\mcA_{\mathrm{rot}}$ is nonzero.

The unitary $z$ commutes with $u,v$ and specifies their commutation factor. In a finite-dimensional representation, each eigenvalue of $z$ must be a root of unity. The last relation then forces both $u$ and $v$ to be the identity. On the other hand, when $2z+2z^*-\Id=0$, it leaves an irrational rotation algebra as a quotient.

\begin{proposition}
  \label{prop:rotation-traces}
  The character $\tau_0$ is the unique finite-dimensional tracial state on $\mcA_{\mathrm{rot}}$. There is also an amenable tracial state $\tau_1\ne\tau_0$ on $\mcA_{\mathrm{rot}}$.
\end{proposition}

\begin{proof}
  Let $\pi:\mcA_{\mathrm{rot}}\arr M_q$ be a unital $*$-representation. Since $z$ commutes with $u,v$, every eigenspace of $\pi(z)$ is invariant under $\pi(u),\pi(v)$. On a nonzero eigenspace of dimension $r$, write $\pi(z)=\mu\Id_r$ and denote the other restrictions by $U,V$. The identity $UV=\mu VU$ gives $\mu^r=1$ by taking determinants. Moreover, $2\mu+2\overline\mu-1\ne0$: otherwise $\mu+\mu^{-1}=1/2$, which is impossible since a rational algebraic integer is an integer. The last relation in \eqref{eq:rot-presentation} therefore gives $U+V=2\Id_r$. It follows that
  \begin{equation*}
    (U-\Id_r)^*(U-\Id_r)+(V-\Id_r)^*(V-\Id_r)=0,
  \end{equation*}
  so $U=V=\Id_r$, and $UV=\mu VU$ now gives $\mu=1$. Thus every finite-dimensional representation sends all three generators to the identity. Applying this to the GNS representation of a finite-dimensional tracial state proves uniqueness of $\tau_0$.

  For the second trace, put $\zeta:=(1+i\sqrt{15})/4$. Then $|\zeta|=1$ and $2\zeta+2\overline\zeta-1=0$. The same algebraic-integer argument shows that $\zeta$ is not a root of unity. The irrational rotation algebra
  \begin{equation*}
    A_\zeta:=C^*\langle U,V\text{ unitary}\mid UV=\zeta VU\rangle
  \end{equation*}
  is nonzero and nuclear and has a tracial state $\tau_\zeta$~\cite[Section~2.2]{BCHL18}. Its trace is therefore amenable. The assignments
  \begin{equation*}
    \theta(z)=\zeta\Id,\quad \theta(u)=U,\quad \theta(v)=V
  \end{equation*}
  satisfy \eqref{eq:rot-presentation} and define a surjective $*$-homomorphism $\theta:\mcA_{\mathrm{rot}}\arr A_\zeta$. By \Cref{lemma:pre_traces}, amenability is preserved under pullback, so $\tau_1:=\tau_\zeta\circ\theta$ is amenable. Finally, $\tau_0\neq \tau_1$ because $\tau_1(z)=\zeta\ne1=\tau_0(z)$.
\end{proof}

We now apply \Cref{thm:projection-construction} using precisely the four relations in \eqref{eq:rot-presentation}. Let $p_{\mathrm{rot}}$ be the resulting correlation defined using \Cref{eq:general-nonrobust-correlation}.

\begin{theorem}
  \label{thm:rotation-nonrobust}
  The correlation $p_{\mathrm{rot}}$ is an extreme synchronous quantum correlation with $250$ questions and two answers per party. It is a self-test with an ideal strategy using a maximally entangled state of Schmidt rank $36$, but it is not a robust self-test.
\end{theorem}

\begin{proof}
  The self-testing assertions follow from \Cref{prop:rotation-traces} and \Cref{thm:nonrobust-from-traces}. We calculate the sizes for the construction in the preceding section. Expanding \eqref{eq:rot-presentation} gives the monomial set
  \begin{equation*}
    \mcS=\{\Id,z,z^*,u,v,zu,uz,zv,vz,uv,z^*u,z^*v,zvu\}.
  \end{equation*}
  The set $\mcW$ adds only $vu$, so $|\mcS|=13$ and $|\mcW|=14$. Thus $k=27$ and \eqref{eq:pc-matrix-size} gives $N=36$. There are $|\mcW|-1+|\mcS|=26$ equations in $\mcI$.

  Assign $u_1=z$, $u_2=u$, and $u_3=v$ in the construction. The common family has
  \begin{equation*}
    |\mcP_0|=36+2(36-1)+2(3+1)+1=115.
  \end{equation*}
  The equations in $\mcI$ request $52$ different connecting projections $C_{a,b}$. More explicitly, the numbers requested with one index $2,3,4,5,6,7,8,9$ are, respectively, $13,13,7,7,3,3,3,3$; their other indices are all of the form $a_w$ or $b_w$. Hence none of these requests repeat or give the initial $C_{2,3}$. Each of the $26$ equations consequently contributes two $G,K$ projections, two $C$ projections, and one $T$ projection. Together with $E$ and the four $F_j$, this gives
  \begin{equation*}
    |X|=115+5\cdot26+1+4=250.
  \end{equation*}
  As throughout the construction, these are counts of the indexed projections; no further identifications are made.
\end{proof}

This proves \Cref{thm:main}. The sizes just calculated belong to this presentation and this projection construction; no minimality is claimed.

\subsection{Lower bounds for small scenarios}
\label{sec:lower}
For binary synchronous quantum correlations that are self-test but not robust, we give lower bounds on the number of questions and the Schmidt rank of the ideal state.

\begin{proposition}
  \label{prop:four-question-lower-bound}
  Let $p\in C_q^s(X,\Z_2)$ be an extreme synchronous quantum correlation. If $p$ is a self-test but not a robust self-test, then $|X|\geq4$, and every ideal strategy for $p$ employs a maximally entangled state of Schmidt rank at least three.
\end{proposition}

\begin{proof}
  We first show that $p$ cannot be realized by projections in $\C$ or $M_2$, with their normalized traces. A realization in $\C$ makes $p$ deterministic. Every implementing tracial state then sends each canonical projection to its prescribed scalar in the tracial GNS quotient, so the implementing trace is unique. By \Cref{thm:tracial-self-testing}, $p$ would be robust.

  Now suppose that projections $P_x\in M_2$ satisfy
  \begin{equation*}
    p(1,1|x,y)=\tr_2(P_xP_y)
  \end{equation*}
  for all $x,y\in X$. Let $X_0:=\{x\in X:P_x\notin\{0,\Id_2\}\}$ and let $p_0$ be the restriction of $p$ to $X_0\times X_0$. We may assume $X_0\ne\varnothing$, since otherwise $p$ is deterministic. The omitted questions have fixed answers. Any convex decomposition of $p_0$ extends to one of $p$ by restoring these answers, so $p_0$ is extreme. Moreover, an implementing tracial state of $p$ is determined by its restriction to $X_0$: each omitted generator equals its prescribed scalar in the tracial GNS quotient. Since restriction preserves amenability, \Cref{thm:tracial-self-testing} shows that robustness of $p_0$ implies robustness of $p$.

  For $x\in X_0$, the projection $P_x$ has rank one. Thus $A_x:=2P_x-\Id_2$ is a traceless self-adjoint unitary, and we may write $A_x=v_x\cdot\sigma$ for a unit vector $v_x\in\R^3$, where $\sigma=(\sigma_1,\sigma_2,\sigma_3)$ is the triple of Pauli matrices. The marginals of $p_0$ are $1/2$, and its correlator matrix is
  \begin{equation*}
    C_{x,y}:=4p_0(1,1|x,y)-1 =\tr_2(A_xA_y)=\langle v_x,v_y\rangle
  \end{equation*}
  for all $x,y\in X_0$.

  The transposed projections $P_x^{T}$ also define a finite-dimensional tracial state implementing $p$, since $\tr_2(P_x^{T}P_y^{T})=\tr_2(P_xP_y)$. By \Cref{thm:tracial-self-testing}, this trace equals the one defined by the $P_x$. Applying this equality to $(2e_{x,1}-\Id)(2e_{y,1}-\Id)(2e_{z,1}-\Id)$ and using the Pauli identities gives
  \begin{equation*}
    \begin{aligned}
      i\det(v_x,v_y,v_z) &=\tr_2(A_xA_yA_z)\\
      &=\tr_2(A_zA_yA_x) =-i\det(v_x,v_y,v_z)
    \end{aligned}
  \end{equation*}
  for all $x,y,z\in X_0$. Hence the vectors $v_x$ span a space of dimension at most two. If this dimension were one, $p_0$ would be the equal mixture of two opposite deterministic assignments, contradicting extremality. Therefore $\operatorname{rank}C=2$, and \Cref{thm:unbiased-even-rank} makes $p_0$, and hence $p$, a robust self-test. This rules out a realization in $M_2$.

  Suppose now that $|X|\leq3$. The representations in the proof of \cite[Theorem~3.15]{Rus20} express every synchronous binary correlation on three questions as a convex combination of correlations realized by projections in $\C$ or $M_2$, with their normalized traces. For fewer questions, first adjoin deterministic questions and then restrict these realizations. Extremality therefore gives such a realization for $p$ itself, which we have just ruled out. Thus $|X|\geq4$.

  Finally, \Cref{thm:ideal-strategy-from-trace} gives an ideal strategy for $p$ with shared state $\ket{\varphi_d}$ and projections in $M_d$, whose normalized trace implements $p$. The preceding argument forces $d\geq3$. Since this ideal strategy is a local dilation of every strategy realizing $p$, and local isometries preserve Schmidt rank, the shared state of every such strategy has Schmidt rank at least $d$. In particular, this holds for every ideal strategy.
\end{proof}

\section{Further consequences of the encoding theorem}
\label{app:further-consequences}

Theorem~\ref{thm:projection-construction} transfers properties of tracial states on finitely presented $C^*$-algebras to properties of individual synchronous binary correlations. We establish a general correspondence, derive the non-closedness of the finite-dimensional and spatial quantum correlation sets, and recover the existence of extreme quantum correlations with inequivalent realizations. We also explain how the absence of embeddable traces yields a separation between quantum-approximate and commuting correlations.

We use $C_{qs}$ for correlations arising from tensor-product strategies on possibly infinite-dimensional Hilbert spaces, and $C_{qc}$ for correlations arising from commuting measurements on a common Hilbert space. A superscript $s$ denotes the synchronous subset. In particular, $C_q^s=C^s_{qs}$~\cite[Theorem~3.10]{KPS18}, so there is no need to discuss synchronous quantum-approximate correlations separately.

Throughout this section, let $\mcA$ have a finite real-coefficient unitary presentation and satisfy $T(\mcA)\ne\varnothing$. Recall that we use $T(\mcA)$ to denote the set of tracial states on $\mcA$. We also use $T_{\mathrm{fin}}(\mcA)$ and $T_{\mathrm{emb}}(\mcA)$ for finite-dimensional tracial states and Connes-embeddable\footnote{A
tracial state is Connes-embeddable if its GNS von Neumann algebra admits a
trace-preserving embedding into an ultrapower of the hyperfinite
$\mathrm{II}_1$ factor.} tracial states, respectively.
%Let $\mathcal R$ denote the hyperfinite $\mathrm{II}_1$ factor, and fix a free ultrafilter $\omega$ on $\mathbb N$. A tracial state $\tau$ on a separable unital $C^*$-algebra $\mcB$ is \emph{Connes-embeddable} if its tracial GNS von Neumann algebra $\pi_\tau(B)''$, equipped with its canonical trace, admits a trace-preserving embedding into $\mathcal R^\omega$. 

Apply \Cref{thm:projection-construction} to $\mcA$, and let $N$, $X$, $\{P_x\}_{x\in X}$, and $\Phi:\PVM{X,\Z_2}\longrightarrow M_N(\mcA)$ the the resulting objects. Let $P_{x,1}:=P_x$ and $P_{x,0}:=1-P_x$. By part~(b) of \Cref{thm:projection-construction}, the value
\begin{equation}
  \label{eq:app-associated-correlation}
  p_\mcA(a,b\mid x,y) := (\tr_N\otimes\tau)(P_{x,a}P_{y,b})
\end{equation}
is independent of $\tau$ in $T(\mcA)$. This gives a synchronous correlation $p_\mcA$ in $C_{qc}(X,\Z_2)$~\cite{PSSTW16}. We refer to $p_\mcA$ as the associated synchronous correlation of $\mcA$. We use $\mcF(p_{\mcA})$ for the tracial states on $\PVM{X,\Z_2}$ implementing $p_{\mcA}$, and
$\mcF_{\mathrm{fin}}(p_{\mcA})$ and
$\mcF_{\mathrm{am}}(p_{\mcA})$ for its finite-dimensional and
amenable tracial states, respectively.

\begin{proposition}
\label{prop:app-trace-correspondence}
The mapping $\Theta:T(\mcA)\arr\mcF(p_{\mcA})$ given by
\begin{equation*}
    \Theta(\tau):=(\tr_N\otimes\tau)\circ\Phi
\end{equation*}
is an affine homeomorphism for the weak-$*$ topologies. Its
restrictions give bijections
$T_{\mathrm{fin}}(\mcA)\arr\mcF_{\mathrm{fin}}(p_{\mcA})$ and
$T_{\mathrm{emb}}(\mcA)\arr\mcF_{\mathrm{am}}(p_{\mcA})$.
Moreover, $p_{\mcA}$ is extreme in
$C_{qc}(X,X,\Z_2,\Z_2)$. In particular:
\begin{enumerate}[(a)]
    \item $p_{\mcA}\in C_q^s(X,\Z_2)$ if and only if
    $T_{\mathrm{fin}}(\mcA)\ne\varnothing$.
    Likewise, $p_{\mcA}\in C_{qa}^s(X,\Z_2)$ if and only if
    $T_{\mathrm{emb}}(\mcA)\ne\varnothing$.

    \item $p_{\mcA}$ belongs to $C_q^s(X,\Z_2)$ and is a self-test
    if and only if $\abs{T_{\mathrm{fin}}(\mcA)}=1$.
    It belongs to $C_q^s(X,\Z_2)$ and is a robust self-test if and
    only if $\abs{T_{\mathrm{fin}}(\mcA)}
            =\abs{T_{\mathrm{emb}}(\mcA)}=1$.
\end{enumerate}
\end{proposition}

\begin{proof}
\Cref{thm:projection-construction}(b)--(c) gives the affine bijection,
since $e_{x,0}=\Id-e_{x,1}$ determines the remaining outcome
probabilities. The map is weak-$*$ continuous. Since $T(\mcA)$ is
compact and $\mcF(p_{\mcA})$ is Hausdorff, it is a homeomorphism.

For every $\tau\in T(\mcA)$, surjectivity of $\Phi$ gives a
trace-preserving isomorphism
\begin{equation*}
    \pi_{\Theta(\tau)}(\PVM{X,\Z_2})''
        \cong M_N\bigl(\pi_\tau(\mcA)''\bigr).
\end{equation*}
Finite dimensionality and Connes embeddability are preserved and
reflected by finite matrix amplification. On $\PVM{X,\Z_2}$,
Connes-embeddable tracial states are exactly the amenable tracial
states; see \cite[Theorem~3.2 and the proof of Theorem~3.6]{KPS18}.
This proves both restricted bijections.

To prove extremality, let $I:=\ker\Phi$. By
\Cref{thm:projection-construction}(b)--(c),
\begin{equation*}
    \mcF(p_{\mcA})
        =\{\varphi\in T(\PVM{X,\Z_2}):\varphi|_I=0\}.
\end{equation*}
This is a face: if a strict convex combination of two tracial states
vanishes on $I$, positivity forces both states to vanish on $c^*c$
for every $c\in I$, and Cauchy--Schwarz then makes them vanish on $I$.
Suppose $p_{\mcA}=tq_1+(1-t)q_2$ with $0<t<1$ and
$q_1,q_2\in C_{qc}(X,X,\Z_2,\Z_2)$.
Nonnegativity of the probabilities makes both $q_i$ synchronous.
Choose implementing tracial states $\varphi_i$ on
$\PVM{X,\Z_2}$~\cite{PSSTW16}.
Their convex combination belongs to $\mcF(p_{\mcA})$, so the face
property gives $\varphi_1,\varphi_2\in\mcF(p_{\mcA})$.
Hence $q_1=q_2=p_{\mcA}$.

Part~(a) follows from the two restricted bijections,
\Cref{thm:synchronous-quantum-fd-trace}, and
\cite[Theorem~3.6]{KPS18}.
Whenever $p_{\mcA}\in C_q^s(X,\Z_2)$, its extremality in $C_{qc}$
also gives extremality in $C_q$. Thus
\Cref{thm:tracial-self-testing} and the same bijections give part~(b).
Here $T_{\mathrm{fin}}(\mcA)\subseteq T_{\mathrm{emb}}(\mcA)$,
so nonemptiness of $T_{\mathrm{fin}}(\mcA)$ and uniqueness in
$T_{\mathrm{emb}}(\mcA)$ are equivalent to the two cardinalities
being one.
\end{proof}

\subsection{An infinite-dimensional self-test}
\label{app:nonclosure}

The correspondence immediately reduces non-closedness to the existence of
an algebra with a Connes-embeddable trace and no finite-dimensional trace.
A modification of the rotation-algebra construction in
\Cref{sec:rotation-example} gives such an algebra. Let
\begin{equation}
    \mcB:=C^*\left\langle z,u,v,t\text{ unitary}\ \middle|\
    \begin{gathered}
        zu=uz,\quad zv=vz,\quad uv=zvu,\quad 2z+2z^*=\Id,\\
        t^2=\Id,\quad tzt=z^*,\quad tut=v
    \end{gathered}
    \right\rangle.
    \label{eq:app-nonclosure-algebra}
\end{equation}
This is a finite real-coefficient unitary presentation. We will show that
$\mcB$ admits a tracial state, so \Cref{thm:projection-construction}
applies. Denote its associated synchronous correlation by $p_{\mcB}$.

\begin{proposition}
\label{prop:app-nonclosure-selftest}
The correlation $p_{\mcB}$ satisfies
\begin{equation*}
    p_{\mcB}\in C_{qa}^s(X,\Z_2)\setminus C_q^s(X,\Z_2)
\end{equation*}
and is extreme in $C_{qc}(X,X,\Z_2,\Z_2)$.
It is a commuting-operator self-test and a quantum-approximate self-test, with no finite-dimensional realization.
\end{proposition}

In the abstract-state formulation of \cite[Section~7]{PSZZ24}, a
correlation $p$ is a \emph{commuting-operator self-test} if it has a unique implementing state on
$\POVM{X,\Z_2}\otimes_{\max}\POVM{X,\Z_2}$; it is a \emph{quantum-approximate self-test} if it has a unique implementing state on
$\POVM{X,\Z_2}\otimes_{\min}\POVM{X,\Z_2}$.
For synchronous correlations, the state--trace correspondence identifies
these conditions with uniqueness of the implementing tracial state and uniqueness of the amenable implementing tracial state, respectively.

\begin{proof}
Let $\zeta:=(1+i\sqrt{15})/4$, which is not a root of unity by the
argument in \Cref{prop:rotation-traces}, and let
\begin{equation*}
    A_\zeta:=C^*\langle U,V\text{ unitary}\mid UV=\zeta VU\rangle.
\end{equation*}
This irrational rotation algebra is nuclear and has a unique tracial
state; see \cite[Section~2.2]{BCHL18}. It has no nonzero
finite-dimensional representation, by the determinant argument in
\Cref{prop:rotation-traces}. We claim that
\begin{equation}
    \mcB\cong M_2(A_\zeta).
    \label{eq:app-monotracial-algebra}
\end{equation}
The assignments
\begin{equation*}
    \begin{aligned}
        z&\longmapsto
            \begin{pmatrix}\zeta\Id&0\\0&\overline\zeta\Id\end{pmatrix},
        &t&\longmapsto\begin{pmatrix}0&\Id\\\Id&0\end{pmatrix},\\
        u&\longmapsto\begin{pmatrix}U&0\\0&V\end{pmatrix},
        &v&\longmapsto\begin{pmatrix}V&0\\0&U\end{pmatrix}
    \end{aligned}
\end{equation*}
satisfy the relations and generate $M_2(A_\zeta)$, so they define a
surjective $*$-homomorphism $\theta:\mcB\arr M_2(A_\zeta)$.
In $\mcB$, define
\begin{equation*}
    e:=\frac{z-\overline\zeta\Id}{\zeta-\overline\zeta}
    ~\text{ and }~ f:=\Id-e.
\end{equation*}
The relation $2z+2z^*=\Id$ gives
$(z-\zeta\Id)(z-\overline\zeta\Id)=0$, so $e,f$ are orthogonal
projections. Also, $tet=f$, and hence $e,et,te,f$ form matrix units.
The relations $tz=z^*t$, $tu=vt$, and $tv=ut$ allow every occurrence
of $t$ in a monomial to be moved to the right. Since $ete=0$,
$e\mcB e$ is generated by the unitaries $eu,ev$, with identity $e$
and relation $(eu)(ev)=\zeta(ev)(eu)$.
They define a surjective $*$-homomorphism $A_\zeta\arr e\mcB e$.
The restriction of $\theta$ to $e\mcB e$, identified with its
$(1,1)$ entry, is its inverse. Thus $e\mcB e\cong A_\zeta$, and
the matrix units give \eqref{eq:app-monotracial-algebra}.

Consequently, $\mcB$ has a unique tracial state $\tau_{\mcB}$,
which is amenable, and it has no finite-dimensional tracial state.
Since amenable traces are Connes-embeddable,
\Cref{prop:app-trace-correspondence} gives
\begin{equation*}
    p_{\mcB}\in C_{qa}^s(X,\Z_2)\setminus C_q^s(X,\Z_2)
\end{equation*}
and extremality in $C_{qc}(X,X,\Z_2,\Z_2)$. It also follows that
\begin{equation*}
    \mcF(p_{\mcB})=\mcF_{\mathrm{am}}(p_{\mcB})
        =\{(\tr_N\otimes\tau_{\mcB})\circ\Phi\}.
\end{equation*}
Thus both self-testing conditions hold. The unique implementing trace
is not finite-dimensional, since $\mcF_{\mathrm{fin}}(p_{\mcB})$
is empty.
\end{proof}

This gives another example of the non-closedness of $C_q$ established
in \cite{Slofstra19,KPS18,DPP19}. The construction has binary outputs,
but does not optimize the number of questions; in particular, it does
not reproduce the five-input bound of \cite{DPP19}.
The infinite-dimensional self-testing property is a byproduct of the
uniqueness of the trace on $\mcB$. Such infinite-dimensional self-tests
also appear in concurrent work by Paddock and Schmidt~\cite{PS26},
Kar, Schafhauser, and Zhao~\cite{KSZ}, and Penner, Slofstra, and
Zhao~\cite{PSZ}.

\subsection{Extreme correlations with inequivalent realizations}
\label{app:inequivalent-realizations}

The existence of extreme quantum correlations with physically inequivalent realizations was established in~\cite{TFRBK21}; see also the discussion in~\cite{MancinskaSchmidt23}. Theorem~\ref{thm:projection-construction} recovers this phenomenon from a finite-dimensional commutative algebra.

\begin{corollary}[Extremality without self-testing]
  \label{cor:app-extreme-not-self-test}
  There is a synchronous binary quantum correlation that is extreme in the full commuting correlation set and has two finite-dimensional realizations, neither of which is a local dilation of the other. This correlation is not a self-test.
\end{corollary}

\begin{proof}
  Apply Theorem~\ref{thm:projection-construction} to
  \[
    \mcA_0 = C^*\langle w\ \text{unitary}\mid w^2=1\rangle \cong\mathbb C\oplus\mathbb C.
  \]
  Let $\chi_+$ and $\chi_-$ be its two characters, determined by $\chi_\pm(w)=\pm1$. Define $\rho_\pm := (\operatorname{id}_{M_N}\otimes\chi_\pm)\circ\Phi$ from $\PVM{X,\Z_2}\longrightarrow M_N(\mathbb C)$ and let $F_{x,a}^\pm:=\rho_\pm(e_{x,a})$. Both $\rho_+$ and $\rho_-$ are surjective.

  Let
  \[
    |\phi_N\rangle := \frac1{\sqrt N}\sum_{j=1}^N |j\rangle\otimes|j\rangle.
  \]
  The two strategies
  \[
    S_\pm = \left( |\phi_N\rangle,\, \{F_{x,a}^\pm\}_{x,a},\, \{(F_{y,b}^\pm)^{\mathsf T}\}_{y,b} \right)
  \]
  both implement $p_{\mcA_0}$, because
  \[
    \langle\phi_N| F_{x,a}^\pm\otimes(F_{y,b}^\pm)^{\mathsf T} |\phi_N\rangle = \operatorname{tr}_N(F_{x,a}^\pm F_{y,b}^\pm) = p_{\mcA_0}(a,b\mid x,y).
  \]

  Their implementing traces $\varphi_\pm:=\operatorname{tr}_N\circ\rho_\pm$ are distinct. Indeed, surjectivity of $\Phi$ gives an element $h\in\PVM{X,\Z_2}$ with $\Phi(h)=1_N\otimes w$, and hence
  \[
    \varphi_+(h)=1, \qquad \varphi_-(h)=-1.
  \]
  By \Cref{prop:app-trace-correspondence}, $p_{\mcA_0}$ is extreme in $C_{qc}$. Since it has distinct finite-dimensional implementing traces, Theorem~3.8(a) shows that it is not a self-test.

  For completeness, the two displayed strategies are also inequivalent under local dilations. If one were a local dilation of the other, their state vectors would satisfy
  \[
    (V_A\otimes V_B)|\phi_N\rangle = |\phi_N\rangle\otimes|\kappa\rangle
  \]
  for an auxiliary state $|\kappa\rangle$. Preservation of Schmidt rank forces $|\kappa\rangle$ to be a product vector. Since both local reduced states of $|\phi_N\rangle$ have full rank, the dilation then restricts to local unitary equivalence on the two local supports. The measurement identities would consequently imply unitary equivalence of $\rho_+$ and $\rho_-$. Their normalized traces would be equal, contradicting $\varphi_+\ne\varphi_-$.
\end{proof}

The same argument gives continuous families of inequivalent realizations. For example, take $A=C(\mathbb T)$, the universal unital $C^*$-algebra generated by one unitary; a redundant unitarity relation may be included to put its presentation in the form of Definition~4.1. Its trace space is the space of probability measures on $\mathbb T$. Each evaluation character
\[
  \chi_\lambda(f):=f(\lambda), \qquad \lambda\in\mathbb T,
\]
gives an irreducible $N$-dimensional representation
\[
  \rho_\lambda = (\operatorname{id}_{M_N}\otimes\chi_\lambda)\circ\Phi.
\]
These representations produce the same extreme correlation $p_{\mcA}$, using $|\phi_N\rangle$, while their implementing traces are pairwise distinct. The preceding argument therefore gives uncountably many pairwise inequivalent finite-dimensional realizations of a single extreme synchronous binary correlation.

\printbibliography
\end{document}